\documentclass[journal,doublecolumn,a4paper]{IEEEtran}

\usepackage[english]{babel}
\usepackage[T1]{fontenc}
\usepackage[utf8]{inputenc}
\usepackage{subfig}
\usepackage{multirow}
\usepackage[font = scriptsize, labelfont = bf, labelsep = period, tableposition = top, singlelinecheck = false, justification = centering]{caption}
\usepackage{placeins,verbatim}
\usepackage{amssymb}
\usepackage{amsmath}
\usepackage{amsthm}
\usepackage{latexsym}
\usepackage{wrapfig}
\usepackage{float}
\usepackage{mathtools}
\usepackage{bm}
\usepackage{listings}
\usepackage{textcomp}
\usepackage{setspace}
\usepackage[dvipsnames]{xcolor}
\usepackage{etoolbox}
\usepackage{hyperref}
\usepackage{tikz}
\usetikzlibrary{shapes,snakes,calc}

\newtheorem{definition}{Definition}
\newtheorem{lemma}{Lemma}
\newtheorem{theorem}{Theorem}
\newtheorem{corollary}{Corollary}
\newtheorem{remark}{Remark}

\newif\ifcomment
\commenttrue

\newif\ifcommentLater
\commentLaterfalse

\newcommand{\ie}{\emph{i.e.}, }

\definecolor{dartmouthgreen}{rgb}{0.05, 0.5, 0.06}

\makeatletter
\def\smalloverbrace#1{\mathop{\vbox{\m@th\ialign{##\crcr\noalign{\kern3\p@}%
  \tiny\downbracefill\crcr\noalign{\kern3\p@\nointerlineskip}%
  $\hfil\displaystyle{#1}\hfil$\crcr}}}\limits}
\makeatother

\makeatletter
\def\smallunderbrace#1{\mathop{\vtop{\m@th\ialign{##\crcr
   $\hfil\displaystyle{#1}\hfil$\crcr
   \noalign{\kern3\p@\nointerlineskip}%
   \tiny\upbracefill\crcr\noalign{\kern3\p@}}}}\limits}
\makeatother

\AtBeginEnvironment{matrix}{\setlength{\arraycolsep}{5pt}}

\def\={\coloneqq}							
\def\tpar#1{\left( #1 \right)}    			
\def\qpar#1{\left[ #1 \right]}  			
\def\gpar#1{\left\{ #1 \right\}} 			
\def\scalar#1#2{\langle #1,#2 \rangle}		
\def\abs#1{\left|#1\right|}					
\def\floor#1{\left\lfloor #1 \right\rfloor} 
\def\cb#1{| #1 \rangle}						
\def\rb#1{\langle #1 |}						
\def\pdm#1{\cb{#1} \rb{#1}}					
\def\bbmatrix#1{\qpar{\;\begin{matrix}#1\end{matrix}\;}} 
\def\W#1{\mathbf{W} \!\tpar{#1}}	        	
\def\Wi#1#2{\mathbf{W}_{\!#1} \!\tpar{#2}}		
\def\BMm#1{\Bm_{\tpar{#1}}}					
\def\q#1{\mathfrak{q}_\mathsf{#1}}

\def\a{\alpha}
\def\As{\mathcal{A}}
\def\b{\beta}
\def\Bm{\mathbf{B}}

\def\Cs{\mathcal{C}}							
\def\Csp{\Cs^\perp}							

\def\cPhi{\cb{\Phi}}

\def\F{\mathbb{F}}

\def\Gm{\mathbf{G}}
\def\GC{\Gm_\Cs}							
\def\GCp{\Gm_{\Csp}}						
\def\H{\mathcal{H}}							

\def\Im{\mathbf{I}}
\def\Ls{\mathcal{L}}

\def\N{\mathbb{N}}

\def\Rs{\mathcal{R}}

\def\Xm{\mathbf{X}}

\def\Zm{\mathbf{Z}}

\newcommand{\A}{\mathcal{A}}

\newcommand{\cT}{\ensuremath{\mathcal{T}}}

\newcounter{alp}
\newcounter{ara}
\newcounter{rom}

\newenvironment{alphalist}{\begin{list}{(\alph{alp})\hfill}{\usecounter{alp}
     \topsep0.4ex \labelwidth.6cm \leftmargin.6cm \labelsep0cm
     \rightmargin0cm \parsep0ex \itemsep0ex}}{\end{list}}
\newenvironment{arabiclist}{\begin{list}{(\arabic{ara})\hfill}{\usecounter{ara}
     \topsep0.4ex \labelwidth.6cm \leftmargin.6cm \labelsep0cm
     \rightmargin0cm \parsep0ex \itemsep0ex}}{\end{list}}

\newcommand{\T}{\mbox{$\!^{\sf T}$}}
\title{Quantum Private Information Retrieval from Coded and Colluding Servers}

\author{
  \IEEEauthorblockN{Matteo Allaix \IEEEauthorrefmark{1}, Lukas Holzbaur \IEEEauthorrefmark{2}, Tefjol Pllaha \IEEEauthorrefmark{1}, Camilla Hollanti \IEEEauthorrefmark{1}}
	
	\IEEEauthorblockA{\small \IEEEauthorrefmark{1} Aalto University, Finland. E-mails:
    \{matteo.allaix, tefjol.pllaha, camilla.hollanti\}@aalto.fi
	}\\
	\IEEEauthorblockA{\small \IEEEauthorrefmark{2} Technical University of Munich, Germany. E-mail:  lukas.holzbaur@tum.de
    }
  \thanks{L.~Holzbaur was supported by TU Munich -- Institute for Advanced Study, funded by the German Excellence Initiative and EU 7th Framework Programme under Grant Agreement No. 291763 and the German Research Foundation (Deutsche Forschungsgemeinschaft, DFG) under Grant No. WA3907/1-1. T.~Pllaha was partially supported by Academy of Finland, under the Grant No. 319578. C.~Hollanti and M.~Allaix were supported by the Academy of Finland, under Grants No. 303819 and 318937. C.~Hollanti was also supported  by the TU Munich -- Institute for Advanced Study, funded by the German Excellence Initiative and the EU 7th Framework Programme under Grant Agreement No. 291763, via a Hans Fischer Fellowship.\newline Parts of this paper have been presented at the \emph{2020 IEEE International Symposium on Information Theory (ISIT)}~\cite{allaix2020quantum}.}
}
\IEEEoverridecommandlockouts

\begin{document}

\maketitle

\begin{abstract}
In the classical private information retrieval (PIR) setup, a user wants to retrieve a file from a database or a distributed storage system (DSS)  without revealing the file identity to the servers holding the data. In the quantum PIR (QPIR) setting, a user privately retrieves a classical file by receiving quantum information from the servers. The QPIR problem has been treated by Song \emph{et al.} in the case of replicated servers, both without collusion and with all but one servers colluding. In this paper, the QPIR setting is extended to  account for maximum distance separable (MDS) coded servers. The proposed protocol works for any $[n,k]$-MDS code and $t$-collusion with $t=n-k$. Similarly to the previous cases, the rates achieved are better than those known or conjectured in the classical counterparts. Further, it is demonstrated how the protocol can adapted to achieve significantly higher retrieval rates from DSSs encoded with a locally repairable code (LRC) with disjoint repair groups, each of which is an MDS code. 
\end{abstract}

\section{Introduction}

Private information retrieval (PIR), a problem initially introduced by Chor \emph{et al.} \cite{chor1995private}, enables an user to download a data item from a database without revealing the identity of the retrieved item to the database owner (user privacy). If additionally the user is supposed to obtain no information about any file other than the requested file (server privacy), the problem is referred to as \emph{symmetric} PIR. In recent years, PIR has gained renewed interest in the setting of distributed storage systems (DSSs), where the servers are storing possibly large files. To protect from data loss in the case of the failure of some number of servers, such systems commonly employ either replication (all servers store all files completely) or erasure-correcting codes (each server stores specific linear combinations of symbols of each file). Popular classes of the latter include maximum distance separable (MDS) codes \cite{macwilliams1977theory} and locally repairable codes (LRCs) \cite{huang2013pyramid,gopalan2012locality, kamath2014codes, tamo2014family}. A further extension of the PIR problem considers the case of \emph{collusion}, where the index of the file requested by the user is required to be private even if an (unknown) subset of servers cooperates. The capacity of PIR is known in a variety of settings \cite{sun2017capacity,sun2018capacity,banawan2018capacity,wang2017symmetric,Banawan2019byzantine}, but is still open for coded and colluding servers \cite{freij2017private,Sun2018conjecture}. Progress towards the general coded colluded PIR capacity was recently made in \cite{Holzbaur2019ITW,holzbaur2019capacity}. 

\begin{table*}[t]
    \centering
    \caption{Known capacity results with $n$ servers. For the classical PIR capacities, we report the asymptotic results with respect to the number of files. The acronym SPIR stands for symmetric PIR. The QPIR results were proved for $n=2$,$n=t+1$, and $n=k+t$ servers, respectively. The result in red is a conjectured result, but a protocol achieving that rate was proposed in \cite{freij2017private}. The result in green is proved in this paper.}
    \begin{tabular}{|l|cc|cc|cc|}
    \hline 
    \textsc{Capacities} & PIR & ref. & SPIR & ref. & QPIR & ref. \\
    \hline
    Replicated storage, & \multirow{2}{*}{$1-\frac{1}{n}$} & \multirow{2}{*}{\cite{sun2017replicated}} & \multirow{2}{*}{$1-\frac{1}{n}$}  & \multirow{2}{*}{\cite{sun2018capacity}} & \multirow{2}{*}{1} & \multirow{2}{*}{\cite{song2019capacitymultiple}} \\
    no collusion & & & & & & \\
    \hline
    Replicated storage, & \multirow{2}{*}{$1-\frac{t}{n}$} & \multirow{2}{*}{\cite{sun2017capacity}} & \multirow{2}{*}{$1-\frac{t}{n}$}  & \multirow{2}{*}{\cite{wang2017colluding}} & \multirow{2}{*}{$\geq \frac{2}{t+2}$} & \multirow{2}{*}{\cite{song2019capacitycollusion}} \\
    $t$-collusion & & & & & & \\
    \hline
    $[n,k]$-MDS coded & \multirow{2}{*}{\textcolor{red}{$1-\frac{k+t-1}{n}$}} & \multirow{2}{*}{\cite{holzbaur2019capacity}} & \multirow{2}{*}{$1-\frac{k+t-1}{n}$}  & \multirow{2}{*}{\cite{wang2017linear}} & \multirow{2}{*}{\textcolor{dartmouthgreen}{$\geq \frac{2}{k+t+1}$}} & \multirow{2}{*}{--} \\
    storage, $t$-collusion & & & & & & \\
    \hline
    \end{tabular}
    \label{tab:Capacities}
\end{table*}

The problem of PIR has been also been considered in the quantum communication setting \cite{kerenidis2004quantum, gall2011quantum,giovannetti2008quantum}, where the problem is referred to as \emph{quantum} PIR (QPIR). More recently, Song \emph{et al.} \cite{song2019capacitymultiple,song2019capacitycollusion,song2020} introduced a scheme for a replicated storage system with classical files, where the servers respond to user's (classical) queries by sending quantum systems. The user is then able to privately retrieve the file by measuring the quantum systems. The servers are assumed to share some maximally entangled states, while the user and the servers are not entangled. The non-colluding case was considered in \cite{song2019capacitymultiple}, and was shown to have capacity\footnote{The quantum PIR schemes in \cite{song2019capacitymultiple,song2019capacitycollusion,song2020} and in this work are symmetric, \emph{i.e.}, the user obtains no information about any file except the requested one. However, for the comparison of our rates to the classical setting we will primarily focus on the (higher) non-symmetric rates for the latter, as we consider this to be the more important setting.} equal to one. This is in stark contrast to the classical replicated (asymptotic) PIR  capacity of $1-\frac{1}{n}$ for $n$ servers. The case of QPIR for all but one servers colluding, \emph{i.e.}, $t=n-1$, was considered in \cite{song2019capacitycollusion}, again achieving higher capacity than the classical counterpart. In this case, the QPIR capacity is $\frac{2}{n}$, while classically (and asymptotically) it is $\frac{1}{n}$. In~\cite{song2020}, the authors extend their work \cite{song2019capacitymultiple,song2019capacitycollusion} by considering symmetric QPIR that can resist any $t$ servers colluding. They prove that the $t$-private QPIR capacity is $1$ for $1\leq t \leq n/2$ and $2(n-t)/n$ for $n/2 < t < n$ and they use the \emph{stabilizer formalism}~\cite{Gottesman97} to construct a capacity-achieving protocol. For the reader's convenience, we report some known results on the capacity in Table~\ref{tab:Capacities}.

\subsection{Contributions}

We consider an MDS coded storage system with (classical) files, where the servers respond to user's (classical) queries by sending quantum systems. The user is then able to privately retrieve the file by measuring the quantum systems. The servers are assumed to share some maximally entangled states, while the user and the servers share no entanglement. We generalize the QPIR protocol for replicated storage systems protecting against collusion of all but one servers \cite{song2019capacitycollusion} to the case of $[n,k]$-MDS coded servers and arbitrary $t$-collusion. That is, we extend from the case of $t=n-1$ collusion to $t=n-k$ collusion. Hence,  the protocol of \cite{song2019capacitycollusion} is the special case of $k=1$ in our protocol. This can be seen as trading off collusion protection for reduced storage overhead. The achieved rate $\sim \frac{2}{k+t}$ (cf. Theorem~\ref{thm:rate}) is higher than the conjectured asymptotic rate $\frac{1}{k+t}$ in the classical coded and colluding PIR~\cite{freij2017private}.

The setup is closely related to the classical PIR schemes of \cite{tajeddine2018private,freij2017private}, for more details see Remark~\ref{rem:quantumLimitation}. In the following we set the number of servers to be $n=k+t$. However, note that the scheme can also be  applied to storage systems with $n'>k+t$ servers by performing it on any subset of $k+t$ servers.

Further, we consider LRC coded storage and adapt our QPIR scheme accordingly. The new scheme can achieve a higher rate, which interestingly enough, does not depend on the length of the code (or number of servers). As usual, this comes with a trade-off on collusion resistance.

\section{Basics on PIR and Quantum Computation}

\subsection{Notation}
For any two vectors $u=(u_1,\ldots, u_n),v=(v_1,\ldots,v_n)$ of the same length $n$ we denote their inner product by $\scalar{u}{v}=u_1v_1+\cdots +u_n v_n$. We will denote by $\qpar{n}$ the set $\gpar{1,2,\ldots,n}$, $n\in\mathbb{N}$, and by $\F_q$ the finite field of $q$ elements. In this paper, we only consider characteristic two, \emph{i.e.}, $q$ is a power of two.
We denote a linear code of length $n$ and dimension $k$ by $[n,k]$. 

Let $\Gamma = \{\gamma_1,....,\gamma_{\mu}\}$ be a basis of $\F_{q^\mu}$ over $\F_q$. For $\alpha \in \F_{q^\mu}$ denote by $\varphi$ the bijective, $\F_q$-linear map
\begin{equation}
\begin{split}
\varphi : \F_{q^\mu} & \to \F_q^\mu \\
\alpha = \sum_{i=1}^{\mu} \alpha_i \gamma_i & \mapsto \tpar{\alpha_1,\ldots,\alpha_{\mu}}.
\end{split}
\label{eq:Field isomorphism}
\end{equation}
For a vector $v\in \F_{q^\mu}^n$ we write, by slight abuse of notation, $\varphi(v) = (\varphi(v_1), \varphi(v_2),...,\varphi(v_n))$ for the component-wise application of the mapping.

\subsection{Linear codes and Distributed Data Storage} \label{sec:DSS}

PIR is closely related to distributed data storage, especially when considering PIR from coded storage. Generally, the goal of a DSS is to increase the system reliability, by distributing the data such that a limited number of failed servers does not incur data loss. To achieve this, each file $x^i$ is split into $\beta k$ parts, which are viewed as elements of $\F_q$, arranged in an $\beta \times k$ matrix. Each row is then encoded with a linear code $\Cs$, a subspace of $\F_q^n$, where $n$ is the number of server in the DSS. We denote its dimension by $k$, refer to $n$ as the length of the code $\Cs$, and to an element of $\F_q$ as a symbol, e.g., "a symbol of file $x^i$" refers to one of the elements of $\F_q$ that make up file $x^i$. For short, we say $\Cs$ is an $[n,k]$ code. A matrix $\GC \in \F_q^{k \times n}$ is said to be a \emph{generator matrix} of an $[n,k]$ code $\Cs$ if its $k$ rows are a basis of the linear space $\Cs$. A message $m \in \F_q^k$ is encoded to a codeword $c\in \Cs$ by $c = m \cdot \GC$. This encoding gives $\beta$ codewords $y^i_{b} \in \Cs, b\in [\b]$ per file $x^i$. Server $s$ then stores the $s$-th position of each codeword. An illustration of such a DSS is given in \autoref{fig:DSS}.
\begin{figure}
\centering
\begin{tikzpicture}[thick,scale=0.9, every node/.style={transform shape}]
\path
(4.35,0) node{
	$\bbmatrix{
	x^1_{1,1}  & \cdots & x^1_{1,k}  \\
	\vdots     & \ddots & \vdots     \\
	x^1_{\b,1} & \cdots & x^1_{\b,k} \\
	\vdots     & \vdots & \vdots     \\
	x^m_{1,1}  & \cdots & x^m_{1,k}  \\
	\vdots     & \ddots & \vdots     \\
	x^m_{\b,1} & \cdots & x^m_{\b,k} \\
	} \quad \cdot \GC = \quad
	\bbmatrix{
	y^1_{1,1}  & \cdots & y^1_{1,n}  \\
	\vdots     & \ddots & \vdots     \\
	y^1_{\b,1} & \cdots & y^1_{\b,n} \\
	\vdots     & \vdots & \vdots     \\
	y^m_{1,1}  & \cdots & y^m_{1,n}  \\
	\vdots     & \ddots & \vdots     \\
	y^m_{\b,1} & \cdots & y^m_{\b,n} \\
	}$
}
(-.2,1.1) node[blue]{file 1}
(-.2,-1) node[blue]{file $m$}
(6,-2.3) node[orange]{\small $\text{\rmfamily\scshape server}_1$}
(7.8,-2.3) node[orange]{\small $\text{\rmfamily\scshape server}_n$}
;
\draw[thin,blue,rounded corners=4pt] (0.6,0.25) rectangle (3.1,2);
\draw[thin,blue,rounded corners=4pt] (.6,-.3) rectangle (3.1,-2);
\draw[thin,orange,rounded corners=4pt] (5.6,-2) rectangle (6.3,2);
\draw[thin,orange,rounded corners=4pt] (7.3,-2) rectangle (8.1,2);
\end{tikzpicture}
\caption{Illustration of a DSS storing $m$ files, each consisting of $\b k$ symbols. The matrix $\GC$ is a generator matrix of an $[n,k]$ code $\Cs$.}
\label{fig:DSS}
\end{figure}
As the symbols of each codeword are distributed over the $n$ servers, a server failure now corresponds to a symbol erasure in each codeword. It is well known that a code can correct any $d-1$ erasures, where $d$ denotes the minimum distance 
\begin{equation*}
    d(\Cs) = \min_{\substack{c_1,c_2 \in \Cs \\ c_1 \neq c_2}} d_{\mathsf{H}}(c_1,c_2)
\end{equation*}
and $d_{\mathsf{H}}(c_1,c_2)$ is the Hamming distance between $c_1$ and $c_2$. If the code we are referring to is clear from context, we write simply write $d$ instead of $d(\Cs)$. In the first part of this work, we focus on MDS codes, \emph{i.e.}, $[n,k]$ codes of minimum distance $d=n-k+1$ (for details see \cite[Chapter~11]{macwilliams1977theory}). In Section~\ref{sec:LRC} we consider LRCs, a class of codes that recently attracted a lot of attention \cite{huang2013pyramid,gopalan2012locality,kamath2014codes, tamo2014family}. A formal definition and some more details on these codes are given in the respective section.

For the PIR scheme introduced in the following, we also require the concept of a dual code $\Cs^{\perp}$ of $\Cs$, which is the dual space of $\Cs$ in $\F_q^n$. For MDS codes it is well known (cf. \cite[Chapter 11]{macwilliams1977theory}) that the dual of an MDS code is again an MDS code. Finally, we use the notion of an information set $\mathcal{I} \subset [n]$ of a code, which refers to any subset of positions of a code $\Cs$ for which the subspace obtained by restricting the code $\Cs$ to the positions indexed by $\mathcal{I}$ is of dimension $k$. Intuitively, an information set is a subset of positions of a codeword which allows for a unique recovery of the corresponding message/file.

\subsection{Quantum Computation}\label{sec:QC}
In this section we introduce the notation related to quantum computation to be used later on. For details we refer the reader to~\cite{NC00}.

A \emph{qubit} is a 2-dimensional Hilbert space $\mathcal{H}$ along with a computational basis, that is, a prespecified orthonormal basis $\mathcal{B} = \{|0\rangle,\,|1\rangle\}$. One typically takes $\mathcal{H} = \mathbb{C}^2$. A \emph{state} of $\mathcal{H}$ is a unit vector $|\psi\rangle \in \mathcal{H}$, while $\langle\psi|$ denotes the adjoint of $|\psi\rangle$, that is, $\langle\psi| = |\psi\rangle^\dagger$. Thus, $\langle\phi|\psi\rangle$ and $|\phi\rangle\langle\psi|$ define inner and outer products respectively.

We will work with multiple 2-qubit systems $\mathcal{H} = \mathcal{H}_i \otimes \mathcal{H}_j$, where $i$ and $j$ are labels that we assign to the qubits. In this case the computational basis is $\mathcal{B}^{\otimes 2} = \{|a\rangle\mid a \in \F_2^2\}$, where $|a\rangle = |a_1\rangle\otimes|a_2\rangle$. We will also use the \emph{maximally entangled state}
\begin{equation}
\cPhi = \frac{1}{\sqrt{2}} \tpar{\cb{00} + \cb{11}}.
\label{eq:Ebit}
\end{equation}

For $a,b \in \F_2$, the \emph{Pauli rotations} are defined as
\begin{gather*}
\Xm = \bbmatrix{0 & 1 \\ 1 & 0} = \sum_{i=0}^1 \cb{i + 1} \rb{i} \Rightarrow \Xm^b = \sum_{i=0}^1 \cb{i + b} \rb{i}, \label{eq:Pauli X} \\
\Zm = \bbmatrix{1 & 0 \\ 0 & -1} = \sum_{i=0}^1 (-1)^i \cb{i} \rb{i} \Rightarrow \Zm^a = \sum_{i=0}^1 (-1)^{ai} \cb{i} \rb{i}. \label{eq:Pauli Z}
\end{gather*}

For $a,b\in \F_2$, the \emph{Weyl operator} is defined as
\begin{equation}
\W{a,b} = \Zm^a \Xm^b  = (-1)^{ab} \sum_{j=0}^1 (-1)^{aj} \cb{j + b} \rb{j}.
\label{eq:Weyl operator}
\end{equation}
Notice that, by the bijection given in~(\ref{eq:Field isomorphism}), we can equivalently define the Weyl operator with an element of $\F_4$ as input. We will write $\mathbf{W}_i(a,b)$ when the Weyl operator acts on the qubit with label $i$. With this notation, the following properties of the Weyl operator are well-known and easy to verify:
\begin{align}
\mathbf{W}(a,b)^{\!\dagger}&= (-1)^{ab} \mathbf{W}(a,b), \label{eq:Weyl transpose} \\
\mathbf{W}(a_1, b_1) \mathbf{W}(a_2, b_2) &= (-1)^{a_2 b_1} \mathbf{W}(a_1 + a_2, b_1 + b_2), \label{eq:Weyl sum} \\
\mathbf{W}(a,b)^{\!\dagger}\mathbf{W}(a,b) &= \mathbf{W}(a,b) \mathbf{W}(a,b)^{\!\dagger}= \Im_2, \label{eq:Weyl orthogonality} \\
\mathbf{W}_2(a,b) \cPhi &= (-1)^{ab} \mathbf{W}_1(a,b) \cPhi, \label{eq:Weyl system moving}
\end{align}

\noindent where $\cPhi$ is defined on qubits $\H_1 \otimes \H_2$ and $\Im_2$ is the $2 \times 2$ identity matrix.
The set
\begin{equation*}
\mathcal{B}_{\F_2^2} \coloneqq \gpar{\BMm{a,b} \coloneqq \mathbf{W}_1(a,b) \pdm{\Phi} \mathbf{W}_1(a,b)^{\!\dagger}\mid a,b \in \F_2}
\label{eq:Bell POVM}
\end{equation*}
is a projection-valued measure (PVM), that is, the matrix $\BMm{a,b}$ is a Hermitian projector for any $(a,b) \in \F_2^2$ and satisfies the \emph{completeness equation}~\cite{NC00}. The \emph{Bell measurement} is the measurement defined by the PVM $\mathcal{B}_{\F_2^2}$ described above.

In the following we will require the \emph{two-sum transmission} protocol~\cite{song2019capacitycollusion}. 

\subsubsection*{Two-sum transmission protocol} This protocol allows to send the sum of two pairs of classical bits by communicating two qubits. Suppose that Alice and Bob possess a qubit $\H_A$ and $\H_B$, respectively, and share the maximally entangled state $\cPhi \in \H_A \otimes \H_B$. They would like to send the sum of Alice's information $\tpar{a_1,a_2} \in \F_2^2$ and Bob's information $\tpar{b_1,b_2} \in \F_2^2$ to Carol through two quantum channels. Notice that there is no communication between Alice and Bob, and that neither of them learns data from the other party. The protocol is depicted in \autoref{fig:Two-sum} and is given as follows:
\begin{arabiclist}
	\item Alice and Bob apply the unitaries $\Wi{A}{a_1,a_2}$ on $\H_A$ and $\Wi{B}{b_1,b_2}$ on $\H_B$, respectively;
	\item Alice and Bob send the qubits $\H_A$ and $\H_B$, respectively, over \emph{noiseless} quantum channels;
	\item Carol performs a Bell measurement on the system $\H_A \otimes \H_B$ and obtains $\tpar{a_1 + b_1,a_2 + b_2}$ as the protocol output.
\end{arabiclist}

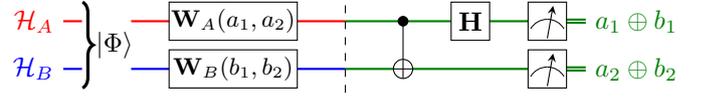
\begin{figure}[ht!]
\centering
\begin{tikzpicture}[scale=1.200000,x=1pt,y=1pt]
\filldraw[color=white] (0.000000, -7.500000) rectangle (163.000000, 22.500000);
\usetikzlibrary{decorations.pathreplacing,decorations.pathmorphing}
\draw[color=red,thick] (0.000000,15.000000) -- (88.000000,15.000000);
\draw[color=black!50!green,thick] (88.000000,15.000000) -- (151.000000,15.000000);
\draw[color=black!50!green,thick] (151.000000,14.500000) -- (163.000000,14.500000);
\draw[color=black!50!green,thick] (151.000000,15.500000) -- (163.000000,15.500000);
\draw[color=red] (0.000000,15.000000) node[left] {$\mathcal{H}_A$};
\draw[color=blue,thick] (0.000000,0.000000) -- (88.000000,0.000000);
\draw[color=black!50!green,thick] (88.000000,0.000000) -- (151.000000,0.000000);
\draw[color=black!50!green,thick] (151.000000,-0.500000) -- (163.000000,-0.500000);
\draw[color=black!50!green,thick] (151.000000,0.500000) -- (163.000000,0.500000);
\draw[color=blue] (0.000000,0.000000) node[left] {$\mathcal{H}_B$};
\draw[fill=white,color=white] (6.000000, -6.000000) rectangle (21.000000, 21.000000);
\draw (16.00000, 7.500000) node {$\cPhi$};
\draw[decorate,decoration={brace,mirror,amplitude = 4.000000pt},very thick] (6.000000,-6.000000) -- (6.000000,21.000000);
\begin{scope}
\draw[fill=white] (53.000000, 15.000000) +(-45.000000:28.284271pt and 8.485281pt) -- +(45.000000:28.284271pt and 8.485281pt) -- +(135.000000:28.284271pt and 8.485281pt) -- +(225.000000:28.284271pt and 8.485281pt) -- cycle;
\clip (53.000000, 15.000000) +(-45.000000:28.284271pt and 8.485281pt) -- +(45.000000:28.284271pt and 8.485281pt) -- +(135.000000:28.284271pt and 8.485281pt) -- +(225.000000:28.284271pt and 8.485281pt) -- cycle;
\draw (53.000000, 15.000000) node {\small{$\mathbf{W}_{A}\!\left( a_1,a_2 \right)$}};
\end{scope}
\begin{scope}
\draw[fill=white] (53.000000, -0.000000) +(-45.000000:28.284271pt and 8.485281pt) -- +(45.000000:28.284271pt and 8.485281pt) -- +(135.000000:28.284271pt and 8.485281pt) -- +(225.000000:28.284271pt and 8.485281pt) -- cycle;
\clip (53.000000, -0.000000) +(-45.000000:28.284271pt and 8.485281pt) -- +(45.000000:28.284271pt and 8.485281pt) -- +(135.000000:28.284271pt and 8.485281pt) -- +(225.000000:28.284271pt and 8.485281pt) -- cycle;
\draw (53.000000, -0.000000) node {\small{$\mathbf{W}_{B}\!\left( b_1,b_2 \right)$}};
\end{scope}
\draw (106.000000,15.000000) -- (106.000000,0.000000);
\filldraw (106.000000, 15.000000) circle(1.500000pt);
\begin{scope}
\draw[fill=white] (106.000000, 0.000000) circle(3.000000pt);
\clip (106.000000, 0.000000) circle(3.000000pt);
\draw (103.000000, 0.000000) -- (109.000000, 0.000000);
\draw (106.000000, -3.000000) -- (106.000000, 3.000000);
\end{scope}
\begin{scope}
\draw[fill=white] (127.000000, 15.000000) +(-45.000000:8.485281pt and 8.485281pt) -- +(45.000000:8.485281pt and 8.485281pt) -- +(135.000000:8.485281pt and 8.485281pt) -- +(225.000000:8.485281pt and 8.485281pt) -- cycle;
\clip (127.000000, 15.000000) +(-45.000000:8.485281pt and 8.485281pt) -- +(45.000000:8.485281pt and 8.485281pt) -- +(135.000000:8.485281pt and 8.485281pt) -- +(225.000000:8.485281pt and 8.485281pt) -- cycle;
\draw (127.000000, 15.000000) node {$\mathbf{H}$};
\end{scope}
\draw[fill=white] (145.000000, 9.000000) rectangle (157.000000, 21.000000);
\draw[very thin] (151.000000, 15.600000) arc (90:150:6.000000pt);
\draw[very thin] (151.000000, 15.600000) arc (90:30:6.000000pt);
\draw[->,>=stealth] (151.000000, 9.600000) -- +(80:10.392305pt);
\draw[fill=white] (145.000000, -6.000000) rectangle (157.000000, 6.000000);
\draw[very thin] (151.000000, 0.600000) arc (90:150:6.000000pt);
\draw[very thin] (151.000000, 0.600000) arc (90:30:6.000000pt);
\draw[->,>=stealth] (151.000000, -5.400000) -- +(80:10.392305pt);
\draw[color=black!50!green] (163.000000,15.000000) node[right] {${a_1 \oplus b_1}$};
\draw[color=black!50!green] (163.000000,0.000000) node[right] {${a_2 \oplus b_2}$};
\draw[dashed] (88.000000, -7.500000) -- (88.000000, 22.500000);
\end{tikzpicture}
\caption{Two-sum transmission protocol: Alice and Bob send the sum of their bits to Carol.}
\label{fig:Two-sum}
\end{figure}
Later on, we will see that, with the aid of isomorphism~\eqref{eq:Field isomorphism}, the two-sum protocol can be used to transmit the sum of values from an extension field.

\subsection{Private Information Retrieval}
Consider a storage system storing $m$ files $x^i, \ i\in[m]$, as described in Section~\ref{sec:DSS}.

In a PIR protocol a user desiring the $K$-th file $x^K$ chooses a query $Q^K=\{Q_1^K,\ldots,Q_n^K\}$ from a query space $\mathcal{Q}$ and transmits $Q_s^K$ to the $s$-th servers.
In the non-quantum PIR setting the response $A_s^K$ from the $s$-th server is a deterministic function of the received query $Q_s^K$ and the shares of the (encoded) files it stores. We denote by $A^K = \{A_1^K,\ldots, A_n^K\}$ the set of responses from all servers. In this work, we consider an extended setting where the user and the servers are also allowed to communicate quantum systems. Briefly, in this QPIR setting, we have $n$ servers and server $s$ possesses qubits that share entanglement with qubits stored by servers $s-1$ and $s+1$, respectively. The total number of pairs which share a maximally entangled state is denoted by $\q{ent}$. Each server applies some standard quantum operations depending on (a function of) the received query and the shares of the (encoded) files it stores to its qubits (\emph{e.g.}, applying a Weyl operator on a qubit or a Bell measurement on a pair of qubits) and \emph{responds} by sending the remaining (non-measured) qubits to the user. The total number of qubits that the servers prepare at the beginning of the protocol is denoted by $\q{in}$, while the total number of qubits that are transmitted from the servers to the user is denoted by $\q{out}$.


\begin{definition}[Correctness]\label{def:correctness}
A QPIR protocol is said to be \emph{correct} if the user can retrieve the desired file $x^K, K\in [m]$ from the responses of the servers.
\end{definition}

As usual, we assume honest-but-curious servers who follow the assigned protocol, but might try to determine the index $K$ of the file desired by the user.
\begin{definition}[Privacy with $t$-Collusion]\label{def:privacy}
  \emph{User privacy:} Any set of at most $t$ colluding servers learns no information about the index $K$ of the desired file.\\
  \emph{Server privacy:} The user does not learn any information about the files other than the requested one.\\
  \emph{Symmetric scheme:} A scheme with both user and server privacy is called  \emph{symmetric}.
\end{definition}

Formally, the QPIR rate in this setting is defined in the following. As customary, we assume that the size of the query vectors is negligible compared to the size of the files.
\begin{definition}[QPIR Rate]\label{def:QPIRrate}
  For a QPIR scheme, \ie a PIR scheme with classical files, classical queries from user to servers and quantum responses from servers to user, the \emph{rate} is the number of retrieved information bits of the requested file per downloaded response qubit, \emph{i.e.},
  \begin{equation*}
    R_{\mathsf{QPIR}} = 
    \frac{\text{\#information bits in a file}}{\text{\#downloaded qubits}} 
    = \frac{\beta k \log_2(q) }{\q{out}} \ .
  \end{equation*}
\end{definition}

For comparison, we also informally define the PIR rate in the non-quantum setting as the number of retrieved information bits of the requested file per downloaded response bit, \emph{i.e.},
\begin{equation}\label{eq:PIRrate}
R_{\mathsf{PIR}} = 
\frac{\text{\#information bits in a file}}{\text{\#downloaded bits}} \ .
\end{equation}
The PIR capacity is the supremum of PIR rates of all possible PIR schemes, for a fixed parameter setting.

\begin{remark}
\label{rem:Qubit capacity}
In this setting we assume that the user does not share any entanglement with the servers. Hence, the maximal number of information bits obtained when receiving a qubit, \emph{i.e.}, the number of bits that can be communicated by transmitting a qubit from a server to the user without privacy considerations, is the binary logarithm of the dimension of the corresponding Hilbert space~\cite{holevo1973bounds}. In this context, with a single qubit we cannot transmit more than 1 bit of classical information.
\end{remark}

\section{$[3,2]$-coded QPIR example with no collusion} \label{sec:32Example}

To provide some intuition on the general scheme presented in Section~\ref{S-col}, we begin with a simple illustrative example.
Let us assume $n=3$ servers which contain, respectively, the pieces $y^i_1 = x^i_1$, $y^i_2 = x^i_2$ and $y^i_3 = x^i_1 + x^i_2$ for $i \in \qpar{m}$. For simplicity, we suppose that $x^i_1,x^i_2 \in \F_2^2$ for each $i \in \qpar{m}$. Note that the storage is MDS-coded (single parity check code). In our QPIR scheme we assume that servers have access to standard quantum resources such as Bell-state preparation and Bell-basis measurement.

\textbf{Preparation Step}. For each $p \in \qpar{2}$ the servers prepare the following qubits and states. Server 2 possesses 4 qubits $\H_2^{\Ls,(p)}$, $\H_2^{\Rs,(p)}$, $\H_2^{(p)}$ and $\H_2^{\As,(p)}$, and the first and the last server possess qubits $\H_1^{(p)}$ and $\H_3^{(p)}$, respectively. Each pair $\tpar{\H_1^{(p)},\H_2^{\Ls,(p)}}, \tpar{\H_2^{\Rs,(p)},\H_3^{(p)}}$ and $\tpar{\H_2^{(p)},\H_2^{\As,(p)}}$ shares a maximally entangled state $\cPhi^{(p)}$. Thus, in this step servers prepare a total of $\q{in} = 2 \cdot (4+1+1) = 12$ qubits and $\q{ent} = 2 \cdot 3 = 6$ pairs which share a maximally entangled state.

The protocol for querying the $K$-th file $x^K$ is depicted in \autoref{fig:QPIR example with PC-code} and is described as follows:
\begin{arabiclist}
	\item Suppose we want to retrieve the piece $y_p^K$ from server $p \in \qpar{2}$. The user generates a uniformly random bit vector $Q^{(p)} = \tpar{Q^{1,(p)}, \ldots,Q^{m,(p)}}$. Let $Q_p^{(p)} \= Q^{(p)} + e_K^m$ and $Q_s^{(p)} = Q^{(p)}$ for any $s \in \qpar{n} \backslash \gpar{p}$, where $Q^{(p)} + e_K^m$ means that we are flipping the $K$-th bit of the query $Q^{(p)}$.
	\item The user sends the query $Q_s^{(p)}$ to server $s$ for each $s \in \qpar{3}$.
	\item Server $s \in \qpar{3}$ computes $H_s^{(p)} = \sum_{i=1}^m Q_s^{i,(p)} y^i_s \in \F_2^2$. The first and the last server apply $\Wi{1}{H_1^{(p)}}$ and $\Wi{3}{H_3^{(p)}}$ to the qubits $\H_1^{(p)}$ and $\H_3^{(p)}$, respectively. Server 2 applies $\Wi{2^\Ls}{H_2^{(p)}}$ to the qubit $\H_2^{\Ls,(p)}$ and performs a Bell measurement on $\H_2^{\Ls,(p)} \otimes \H_2^{\Rs,(p)}$ whose outcome is denoted by $G_2^{(p)} \in \F_2^2$. Then, server 2 applies $\Wi{2}{G_2^{(p)}}$ to the qubit $\H_2^{(p)}$. This last operation initializes the two-sum transmission protocol (cf. Sec.~\ref{sec:QC}) between server 2 and the user: the input qubits $\H_2^{(p)} \otimes \H_2^{\As,(p)}$ are used to send $G_2^{(p)} + \tpar{0,0}$ to the user.
	\item Each server sends its qubit $\H_s^{(p)}$ to the user. Server 2 also sends its additional qubit $\H_2^{\As,(p)}$.
	\item The user performs a Bell measurement to the pair $\H_2^{(p)} \otimes \H_2^{\As,(p)}$ to retrieve $G_2^{(p)} + \tpar{0,0}$ via the two-sum transmission protocol. Then he applies $\Wi{3}{G_2^{(p)}}$ to the qubit $\H_3^{(p)}$ and performs a Bell measurement on $\H_1^{(p)} \otimes \H_3^{(p)}$, whose outcome is $y_p^K$ with probability 1.
	\item Repeat all the previous steps for every piece $p \in \qpar{2}$.
	\item The user reconstructs the desired file $x^K = \tpar{x_1^K,x_2^K} = \tpar{y_1^K,y_2^K}$.
\end{arabiclist}
\begin{figure}[h]
\centering
\includegraphics[width = 0.45\textwidth]{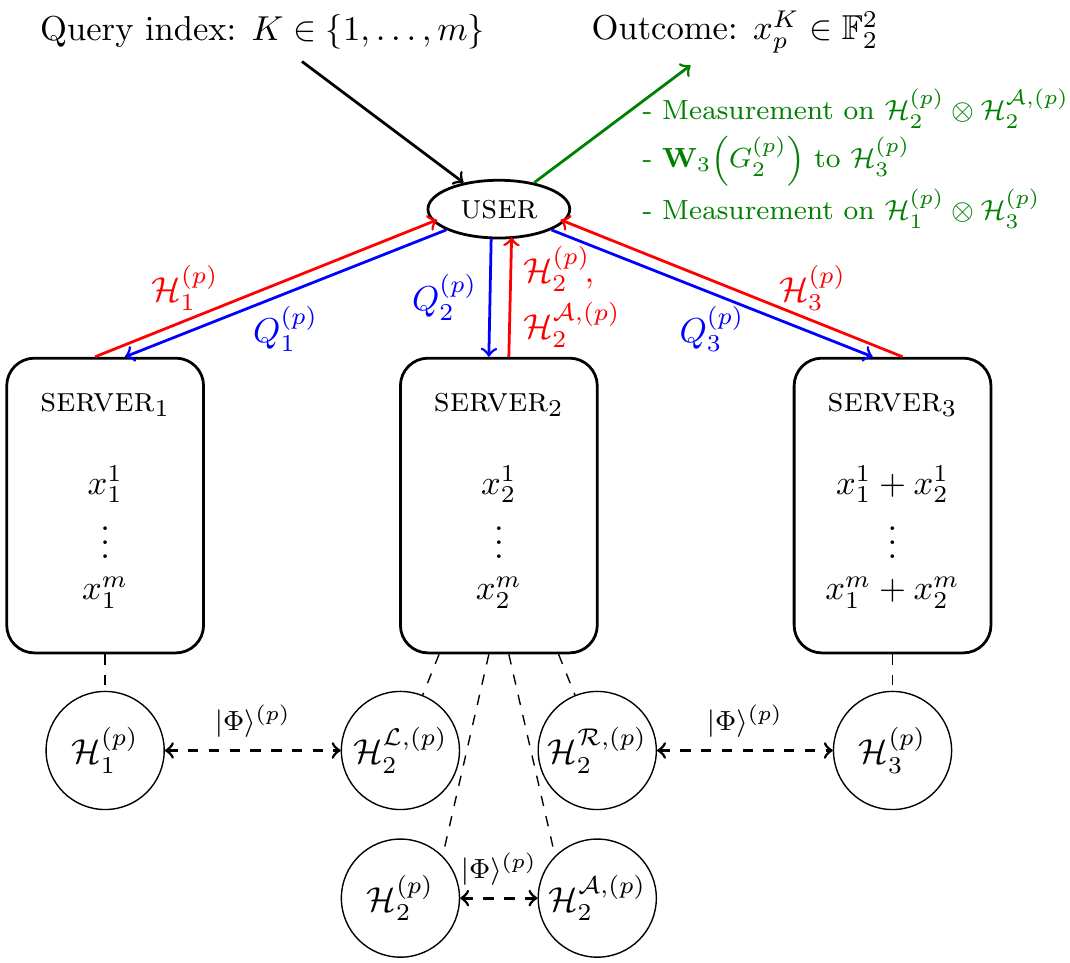}
\caption{QPIR protocol with a [3,2]-parity-check code, as described in Section~\ref{sec:32Example}. Step 2 is depicted in blue, Step 4 in red and Step 5 in green.} 
\label{fig:QPIR example with PC-code}
\end{figure}
\begin{figure}[h]
\centering
\includegraphics[width = 0.4\textwidth]{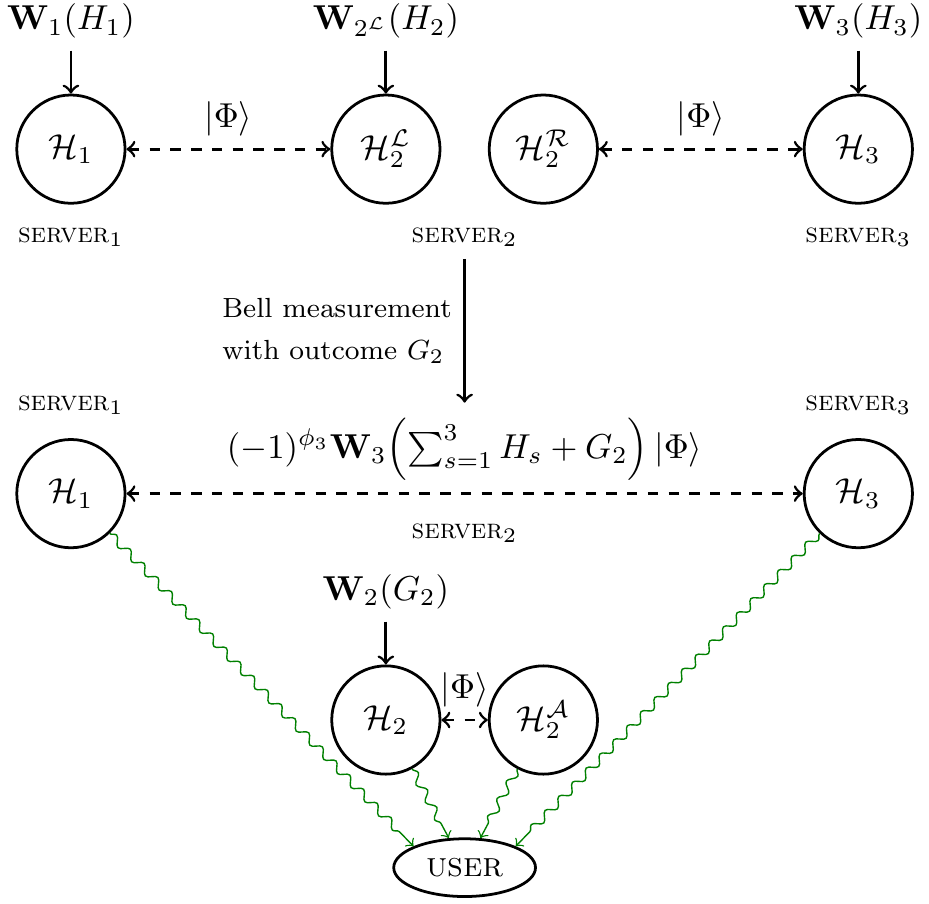}
\caption{Illustration of the operations performed by the servers in the $[3,2]$-coded QPIR example with no collusion, as described in Section~\ref{sec:32Example}.  The download step is depicted in green. These operations are performed for every $p \in \qpar{2}$. Here, we dropped the superscript $(p)$ for a clearer visualization.} 
\label{fig:QPIR operations with PC-code}
\end{figure}

The operations performed by the servers and the download step (Step 4) are visualized in \autoref{fig:QPIR operations with PC-code}. After Step 4, the user possesses the entangled pairs of qubits $\H_2^{(p)} \otimes \H_2^{\As,(p)}$ and $\H_1^{(p)} \otimes \H_3^{(p)}$ for each $p \in \qpar{2}$. The states of those two pairs are, respectively,
\begin{gather*}
\Wi{2}{G_2^{(p)}} \cPhi^{(p)}, \\
\begin{split}
& (-1)^{\phi_3^{(p)}} \Wi{3}{\sum_{s=1} H_s^{(p)} + G_2^{(p)}} \cPhi^{(p)} \\
& = (-1)^{\phi_3^{(p)}} \Wi{3}{y^K_p + G_2^{(p)}} \cPhi^{(p)},
\end{split}
\end{gather*}

\noindent where $\phi_3^{(p)} \in \F_2$ is determined upon $H_1^{(p)}, H_2^{(p)}, H_3^{(p)}$ and $G_2^{(p)}$. The last equality holds for every $p \in \qpar{2}$. In fact, assuming $p=1$,
\[ \begin{split}
\sum_{s=1}^3 H_s^{(1)} & = \sum_{i=1}^m \tpar{Q^i + \delta_{iK}} x^i_1 \! + \! \sum_{i=1}^m Q^i x^i_2 \! + \! \sum_{i=1}^m Q^i \tpar{x^i_1 + x^i_2} \\
& = \sum_{i=1}^m Q^i x^i_1 + x^K_1 + \sum_{i=1}^m Q^i x^i_2 + \sum_{i=1}^m Q^i \tpar{x^i_1 + x^i_2} \\
& = \sum_{i=1}^m Q^i \overbrace{\tpar{x^i_1 + x^i_2 + x^i_1 + x^i_2}}^0 + x^K_1 = x^K_1.
\end{split} \]

\noindent The proof for $p=2$ is the same. Performing a Bell measurement on the first pair, the user retrieves $G_2^{(p)}$ and then applies $\W{G_2^{(p)}}$ to the qubit $\H_3^{(p)}$. Doing so, the state of the second pair becomes
\[ \begin{split}
& (-1)^{\phi_3^{(p)}} \Wi{3}{G_2^{(p)}} \Wi{3}{y^K_p + G_2^{(p)}} \cPhi^{(p)} \\
& \smalloverbrace{=}^{\eqref{eq:Weyl sum}} (-1)^{\phi_u^{(p)}} \mathbf{W}_3 \Bigl( y^K_p + \overbrace{G_2^{(p)} + G_2^{(p)}}^0 \Bigr) \cPhi^{(p)} \\
& = (-1)^{\phi_u^{(p)}} \Wi{3}{y^K_p} \cPhi^{(p)},
\end{split} \]

\noindent where $\phi_u^{(p)} \in \F_2$ is determined upon $\phi_3^{(p)}, G_2^{(p)}$ and $y_p^K + G_2^{(p)}$. Performing now the Bell measurement on the second pair, the user retrieves $y^K_p$ with probability 1.

User and server secrecy follow by Lemma~\ref{lem:Secrecy}. We now notice that a total of 4 qubits are transmitted to the user in each instance of Step 4. Hence, the rate is $R_\mathsf{QPIR} = \frac{2 \cdot 2}{2 \cdot 4} = \frac{1}{2}$, since we recover four bits and download a total of $\q{out} = 8$ qubits over two rounds.

Here, we are not taking advantage of the two-sum transmission protocol in order to send the component $G_2^{(p)}$ during each round from server 2 to the user. For details to a more practical solution, we refer the reader to Remark~\ref{rem:Odd case}.

\section{$\qpar{n,k}$-MDS coded QPIR with $t$-collusion}\label{S-col}

Let $n$ be the number of servers, $L = \min \gpar{l \in \N : 4^l \geq n}$ and $\F_{4^L}$ be the finite field with $4^L$ elements. We present a protocol with user and server secrecy, in which the user retrieves a file $x^K$ from a DSS of $n$ servers coded with an $\qpar{n,k}$-MDS storage code $\Cs$ with base field $\F_{4^L}$. We consider the base field $\F_{4^L}$ for two reasons: first, MDS codes are known to exist for any length $n$ and dimension $k$ with $k\leq n$, if the number of elements in the finite field is at least the length $n$ of the code (cf. \cite[Chapter~10]{macwilliams1977theory}); second, the (encoded) files need to be written as sequences of elements in $\F_4$ via bijection~\ref{eq:Field isomorphism} in order to be transmitted, since the Weyl operator requires a symbol of $\F_4$ as input. We denote the $k \times n$ generator matrix of the code by $\GC$. In this setting, we can protect against $\tpar{n-k}$-collusion. Each file $x^i$ is split into $k$ pieces and then divided into $\b$ stripes $x_b^i = \tpar{x_{b,1}^i,\ldots,x_{b,k}^i}$, \ie the set of files is given by $\mathcal{X} = \left\{x_b^i \in \tpar{\F_{4^L}}^k : i \in \qpar{m}, b \in \qpar{\b} \right\}$ with file size $F = k \b \log_2 \tpar{4^L} = 2 k L \b$. Server $s \in \qpar{n}$ stores the symbols $y^i_{\cdot,s} = \tpar{y^i_{1,s},\ldots,y^i_{\b,s}}$, where $y^i_{b,s}$ is the $s$-th element of the vector $x^i_b \GC$.

Since the storage code is MDS, we need any $k$ symbols of the codeword in order to retrieve the file \cite{macwilliams1977theory}. Without loss of generality, we will retrieve the symbols stored in the first $k$ servers, \ie we will retrieve the symbols $y^K_{\cdot,s}$ from server $s \in \qpar{k}$ after $\b$ rounds.

\subsection{A coded QPIR scheme} 
\label{sec:scheme}

\noindent \textbf{Preparation Step}. For each $p \in \qpar{k}$ the servers prepare the following qubits and states. Server $s \in \gpar{2,\ldots,n-1}$ possesses $3L\b$ qubits $\H_s^{\Ls,(l,b,p)}$, $\H_s^{\Rs,(l,b,p)}$ and $\H_s^{(l,b,p)}$, where $l \in \qpar{L}$, $b \in \qpar{\b}$. The first and the last server possess $L\b$ qubits $\H_1^{(l,b,p)} = \H_1^{\Rs,(l,b,p)}$ and $\H_{n}^{(l,b,p)} = \H_{n}^{\Ls,(l,b,p)}$, respectively. If $n$ is odd, server $n-1$ possesses $L\b$ additional qubits $\H_{n-1}^{\As,(l,b,p)}$. Each pair $\tpar{\H_s^{\Rs,(l,b,p)},\H_{s+1}^{\Ls,(l,b,p)}}$ and $\tpar{\H_{2c}^{(l,b,p)},\H_{2c+1}^{(l,b,p)}}$ shares a maximally entangled state $\cPhi^{(l,b,p)}$ for any $s \in \qpar{n-1}$ and for any $c \in \qpar{\floor{\frac{n}{2}} - 1}$. If $n$ is odd, then an additional $\cPhi^{(l,b,p)}$ is shared between $\H_{n-1}^{(l,b,p)}$ and $\H_{n-1}^{\As,(l,b,p)}$. Thus, if $n$ is even, in this step servers prepare a total of $\q{in} = k (3 L \b (n-2) + 2 L \b) = k L \b (3n - 4)$ qubits and $\q{ent} = k L \b (n-1) + k L \b \tpar{\frac{n}{2} - 1} = k L \b \frac{3n-4}{2}$ pairs which share a maximally entangled state. If $n$ is odd, they prepare a total of $\q{in} = k (3 L \b (n-2) + 3 L \b) = 3 k L \b (n - 1)$ qubits and $\q{ent} = k L \b (n-1) + k L \b \tpar{\frac{n-1}{2} - 1} + k L \b = 3 k L \b \frac{n-1}{2}$ pairs which share a maximally entangled state.

The protocol for querying the $K$-th file $x^K$ is depicted in \autoref{fig:QPIR with MDS-code}  and is described as follows:
\begin{arabiclist}
	\item Suppose the user wants to retrieve the symbols $y^K_{\cdot,p}$ stored in server $p \in \qpar{k}$. Then, he generates $n-k$ independent and uniformly random vectors $Z^{(p)}_1,\ldots,Z^{(p)}_{n-k} \in \tpar{\F_{4^L}}^m$, encodes them as codewords of the dual code of $\Cs$ and adds a 1 in position $K$ to the query directed to server $p$. In other words, the user builds the queries $Q^{(p)}_1,\ldots,Q^{(p)}_n$ multiplying the random vectors by the generator matrix $\GCp$ of the dual code $\Csp$, \ie
	\[
	\bbmatrix{Q^{(p)}_1 & \cdots & Q^{(p)}_n}=\bbmatrix{Z^{(p)}_1 &\cdots & Z^{(p)}_{n-k}}\cdot \GCp +\bm{\xi}_{K,p},
	\]
	where $\bm{\xi}_{K,p}$ is an $m \times n$ matrix whose $\tpar{K,p}$-th entry is 1 and all the other entries are 0.
	\item The user send query $Q_s^{(p)}$ to each server $s\in \qpar{n}$.
	\item In round $b$, server $s \in \qpar{n}$ computes $H_s^{(b,p)} = \scalar{Q^{(p)}_s}{y_{b,s}} \in \F_{4^L}$, where $y_{b,s} = \tpar{y^1_{b,s},\ldots,y^m_{b,s}}$, and divides it into $L$ elements $H_s^{(1,b,p)},\ldots,H_s^{(L,b,p)}$ of $\F_2^2$ by the bijection $\varphi$ defined in~\eqref{eq:Field isomorphism}. For each $l \in \qpar{L}$, the servers perform these steps:
	\begin{alphalist}
		\item server $1$ and server $n$ apply $\W{H_1^{(l,b,p)}}$ and $\W{H_n^{(l,b,p)}}$ to the qubits $\H_1^{(l,b,p)}$ and $\H_n^{(l,b,p)}$, respectively;
		\item server $s \in \gpar{2,\ldots,n-1}$ applies $\W{H_s^{(l,b,p)}}$ to the qubit $\H_s^{\Ls,(l,b,p)}$ and performs a Bell measurement on $\H_s^{\Ls,(l,b,p)} \otimes \H_s^{\Rs,(l,b,p)}$ the outcome of which is denoted by $G_s^{(l,b,p)} \in \F_2^2$. Then, server $s$ applies $\W{G_s^{(l,b,p)}}$ to the qubit $\H_s^{(l,b,p)}$. This last operation initializes the two-sum transmission protocol (cf. Sec.~\ref{sec:QC}) between servers $2c$ and $2c+1$, $c \in \qpar{\floor{\frac{n}{2}} - 1}$, and the user: the input qubits $\H_{2c}^{(l,b,p)} \otimes \H_{2c+1}^{(l,b,p)}$ are used to send $G_{2c}^{(l,b,p)} + G_{2c + 1}^{(l,b,p)}$ to the user. If $n$ is odd, the operation on server $n-1$ initializes the two-sum transmission protocol between server $n-1$ and the user: the input qubits $\H_{n-1}^{(l,b,p)} \otimes \H_{n-1}^{\As,(l,b,p)}$ are used to send $G_{n-1}^{(l,b,p)} + (0,0)$ to the user.
	\end{alphalist}
	\item Each server sends its $L$ qubits $\H_s^{(l,b,p)}$ to the user. If $n$ is odd, server $n-1$ sends its additional $L$ qubits $\H_{n-1}^{\As,(l,b,p)}$.
	\item For each $l \in \qpar{L}$, the user performs the following steps:
	\begin{alphalist}
		\item If $n$ is even, he performs a Bell measurement on each pair $\H_{2c}^{(l,b,p)} \otimes \H_{2c+1}^{(l,b,p)}$ to retrieve $G_{2c}^{(l,b,p)} + G_{2c + 1}^{(l,b,p)}$ via the two-sum transmission protocol for every $c \in \qpar{\frac{n}{2} - 1}$, and computes $G^{(l,b,p)} = \sum_{c=1}^{\frac{n}{2} - 1} \tpar{G_{2c}^{(l,b,p)} + G_{2c + 1}^{(l,b,p)}}$. If $n$ is odd, he also performs a Bell measurement on the pair $\H_{n-1}^{(l,b,p)} \otimes \H_{n-1}^{\As,(l,b,p)}$ to retrieve $G_{n-1}^{(l,b,p)}$ and computes $G^{(l,b,p)} = \sum_{c=1}^{\floor{\frac{n}{2}} - 1} \tpar{G_{2c}^{(l,b,p)} + G_{2c + 1}^{(l,b,p)}} + G_{n-1}^{(l,b,p)}$.
		\item He applies $\W{G^{(l,b,p)}}$ to the qubit $\H_n^{(l,b,p)}$ and performs a Bell measurement on $\H_1^{(l,b,p)} \otimes \H_n^{(l,b,p)}$, whose outcome is $y_{b,p}^{K,(l)}$ with probability 1.
	\end{alphalist}
	Finally, he reconstructs $y_{b,p}^K$ from the $L$ outcomes through the bijection \eqref{eq:Field isomorphism}.
	\item Repeat Steps 3, 4 and 5 for every round $b \in \qpar{\b}$.
	\item Repeat all the previous steps for every piece $p \in \qpar{k}$.
	\item Now the user possesses $\gpar{y_{b,p}^K : b \in \qpar{\b},p \in \qpar{k}}$. First, he reconstructs $y^K_{\cdot,p}$ from the $\b$ elements of $\F_{4^L}$ for each $p \in \qpar{k}$ and builds $y^K = \tpar{y^K_{\cdot,1},\ldots,y^K_{\cdot,k}}$. Then, he computes the desired file $x^K$ from the equation $y^K = x^K \GC'$, where $\GC'$ is the $k \times k$ submatrix of $\GC$ constructed with its first $k$ columns.
\end{arabiclist}

\begin{figure}[h]
\begin{center}
\includegraphics[width = \columnwidth]{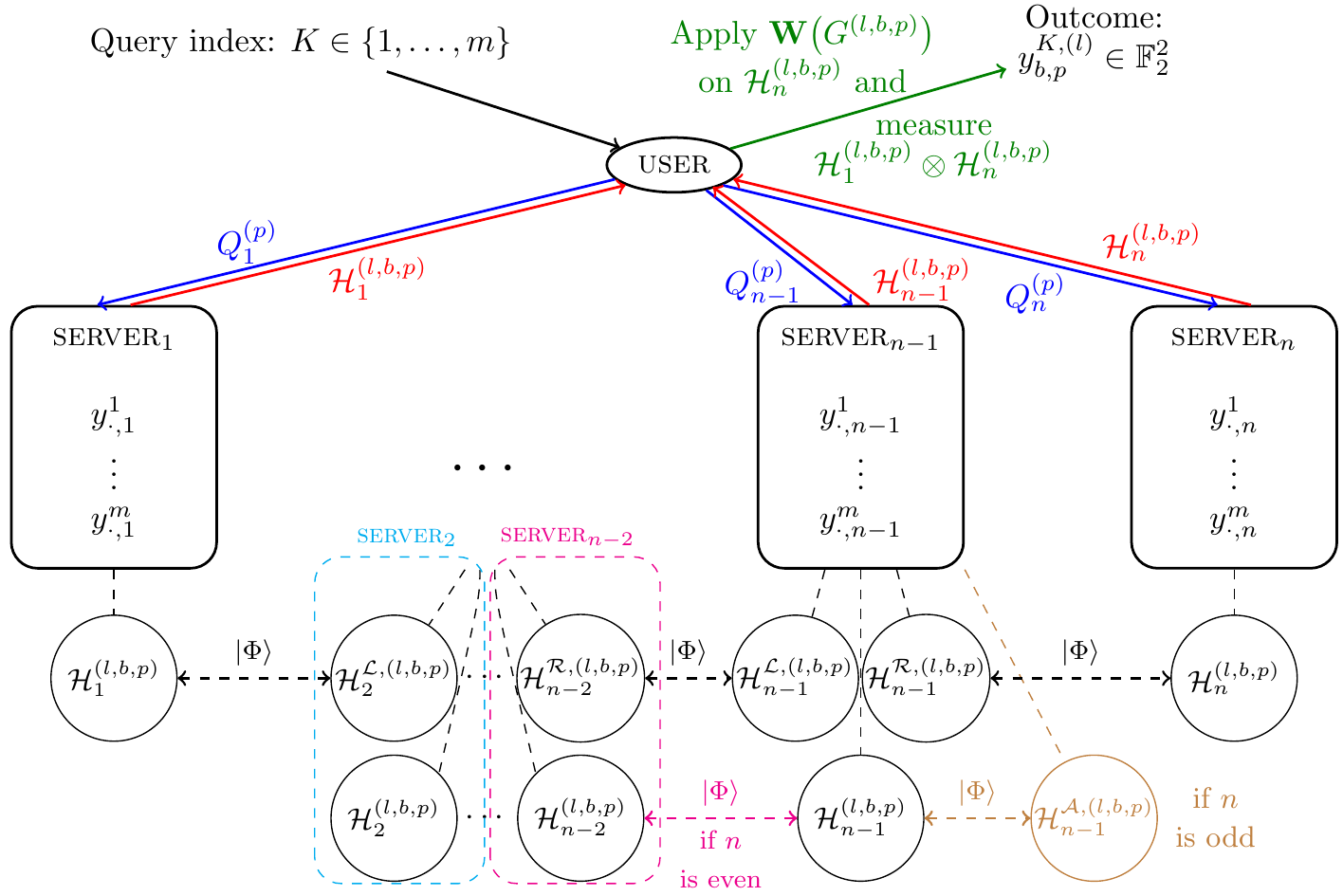}
\end{center}
\caption{QPIR protocol with an $\qpar{n,k}$-MDS code as described in Section~\ref{sec:scheme}. Step 2 is depicted in blue, Step 4 in red and Step 5 in green. Steps 4 and 5 are repeated $L\b$ times per piece. The maximally entangled state $\cPhi$ in the figure denotes $\cPhi^{(l,b,p)}$.} 
\label{fig:QPIR with MDS-code}
\end{figure}

\begin{remark}
In the case of an odd number $n$ of servers, since we are not taking advantage of the two-sum transmission protocol in order to send the components $G_{n-1}^{(l,b,p)}$ during each round from server $n-1$ to the user, we could use classical channels instead of quantum operations and quantum channels to improve the stability of the transmission. We choose to keep the quantum operations in order to have the servers' response completely described by quantum operations. A more practical way to keep the quantum operations, which uses less quantum resources but achieves the same result, would be to simply send the qubits from server $n-1$ each communicating one single bit of information without using the entanglement and the Bell measurement. In this scenario, server $n-1$ prepares the basis state $\cb{G_{n-1}^{(l,b,p)}}$ on $\H_{n-1}^{(l,b,p)} \otimes \H_{n-1}^{\As,(b,p)}$, while the user, once he has received the qubits from the servers, performs a basis measurement on $\H_{n-1}^{(l,b,p)} \otimes \H_{n-1}^{\As,(b,p)}$ to retrieve $G_{n-1}^{(l,b,p)}$. Instead, for ease of notation, we choose to describe the generation of this response consistent with the generation of the other responses. 
\label{rem:Odd case}
\end{remark}

\subsection{Properties of the coded QPIR scheme} \label{sec:Properties of coded QPIR}
\begin{lemma}
The scheme of Section~\ref{sec:scheme} is correct, \emph{i.e.}, fulfills Definition~\ref{def:correctness}.
\end{lemma}
\begin{proof}
The final state before the measurement performed by the user is reached in the same way as the one of the QPIR protocol in \cite{song2019capacitycollusion} for each $p \in \qpar{k}$. Thus, during round $b \in \qpar{\b}$ and for packet $l \in \qpar{L}$, that final state is
\begin{equation}\label{eq:finalState}
(-1)^{\widetilde{\phi}_n} \Wi{n}{\sum_{s=1}^n H_s^{(l,b,p)}} \cPhi,
\end{equation}
where $\widetilde{\phi}_n \in \F_2$ is determined upon $H_1^{(l,b,p)},\ldots,H_n^{(l,b,p)}$, $G_1^{(l,b,p)},\ldots,G_{n-1}^{(l,b,p)}$ and $G^{(l,b,p)}$.

We need to prove that the outcome of the measurement performed on the system $\H_1^{(l,b,p)} \otimes \H_n^{(l,b,p)}$ is $y_{b,p}^{K,(l)}$ for every $b \in \qpar{\b}$, $p \in \qpar{k}$ and $l \in \qpar{L}$. Suppose $L = 1$ and fixed $b,p$. We have $H_s^{(b,p)} = \scalar{Q^{(p)}_s}{y_{b,s}}$. Defining $e_p^n = \tpar{\delta_{1,p},\ldots,\delta_{n,p}}$ where $\delta_{i,j}$ is the Kronecker's delta\footnote{The Kronecker's delta is defined to be $\delta_{i,j}=1$ if $i=j$, $\delta_{i,j}=0$ otherwise.}, $\delta_{i,K} e_p^n$ is the $i$-th row of $\bm{\xi}_{K,p}$. Therefore,
\begin{align*}
& \sum_{s=1}^n H_s^{(b,p)} = \sum_{s=1}^n \scalar{Q^{(p)}_s}{y_{b,s}} = \sum_{s=1}^n \sum_{i=1}^m Q^{i,(p)}_s y_{b,s}^i \\ 
& = \sum_{i=1}^m \sum_{s=1}^n Q^{i,(p)}_s y_{b,s}^i = \sum_{i=1}^m \scalar{Q^{i,(p)}}{\smalloverbrace{\tpar{y_{b,1}^i,\ldots,y_{b,n}^i}}^{y^i_b}} \\
& = \sum_{i=1}^m y^i_b \tpar{Q^{i,(p)}}^{\!\dagger} = \sum_{i=1}^m \tpar{x^i_b \GC} \tpar{Z^{i,(p)} \GCp + \delta_{i,K} e_p^n}^{\!\dagger}\\
& = \sum_{i=1}^m x^i_b \overbrace{\GC \tpar{\GCp}\T}^0 \tpar{Z^{i,(p)}}^{\!\dagger} + x^K_b \GC \tpar{e_p^n}^{\!\dagger} \\ 
& \stackrel{\mathsf{(a)}}{=} y^K_b \tpar{e_p^n}^{\!\dagger} = y^K_{b,p},
\end{align*}
where $\mathsf{(a)}$ holds because the dual of a code is its nullspace.

\noindent If $L > 1$, we have that $H_s^{(l,b,p)}$ is the $l$-th entry of $\scalar{Q^{(p)}_s}{y_{b,s}}$ written as vector of $\tpar{\F_2^2}^L$. Then, by the same arguments as above, we get $\sum_{s=1}^n H_s^{(l,b,p)} = y^{K,(l)}_{b,p}$. As the bijection of \eqref{eq:Field isomorphism} between $\F_{4^L}$ and $\tpar{\F_2^2}^L$ preserves addition, we have that
\[ \begin{split}
\sum_{s=1}^n H_s^{(b,p)} & = \sum_{s=1}^n \tpar{H_s^{(1,b,p)},\ldots,H_s^{(L,b,p)}} \\
& \stackrel{\eqref{eq:Field isomorphism}}{=} \tpar{\sum_{s=1}^n H_s^{(1,b,p)},\ldots,\sum_{s=1}^n H_s^{(L,b,p)}} \\
& = \tpar{y^{K,(1)}_{b,p},\ldots,y^{K,(L)}_{b,p}} = y^K_{b,p}.
\end{split}
\]

The same argument applies to the packetization of the element $G^{(b,p)} = \tpar{G^{(1,b,p)}, \ldots,G^{(L,b,p)}}$. Thus, the correctness of the protocol is proved for every $n \in \N$.
\end{proof}

\begin{remark}
\label{rem:quantumLimitation}
The presented scheme can be viewed as the ``quantum version'' of the PIR scheme for MDS coded storage in \cite{tajeddine2018private}, for those familiar with the schemes of \cite{tajeddine2018private,freij2017private} we shortly elaborate on their connection to the presented scheme here. In \cite{tajeddine2018private} the vector holding the $n$ replies of the servers (denoted $H_s$ in our work) can be viewed as a random codeword of a single parity check (SPC) code plus one symbol of the desired file added in one (known) position. As is the definition of SPCs, summing over all components of this received vector leaves only the symbol of the desired file. This sum is exactly the sum in \eqref{eq:finalState}. In \cite{freij2017private} the scheme of \cite{tajeddine2018private} was generalized to any number of colluding servers $t$ and storage code dimension $k$ with $t+k\leq n$, allowing for the recovery of up to $n-(k+t-1)$ symbols in each round. Similarly, the received vector is a random codeword, but symbols of the desired file are added in multiple positions. To recover each of the desired symbols, the user performs \emph{multiple, distinct linear combinations} of the remaining positions. In this case, for a fixed dimension $k$ and collusion resistance $t$, every additional symbol that is retrieved gives an additional symbol of the desired file, as the same symbols can be used to perform the respective linear combination for each symbol. The property that allows our quantum PIR protocol, as well as the protocols in \cite{song2019capacitymultiple,song2019capacitycollusion}, to increase the rate compared to the classic setting is the fact that the user already receives a linear combination of these elements as an outcome of the quantum measurement, and therefore does not have to download all the symbols $H_s$ separately. However, this quantum measurement entails that only \emph{one single} linear combination can be performed, and thereby only one symbol retrieved per round. This leads to the condition $n-(k+t-1)=1$ or equivalently $n=k+t$. Note that since the restriction of an MDS code to a subset of positions is again an MDS code, we can always apply the protocol to a subset of $k+t$ servers if the actual number of servers in a storage system is $n>k+t$.  
\end{remark}

\begin{lemma} \label{lem:Secrecy}
The scheme of Section~\ref{sec:scheme} is symmetric and protects against  $t$-collusion in the sense of Definition~\ref{def:privacy}.
\end{lemma}
\begin{proof}
Privacy in the quantum part of the protocol follows directly from the privacy of the protocol with all but one servers colluding. For details, we refer the reader to \cite{song2019capacitycollusion}. Intuitively, a colluding set of servers has access to their received queries and quantum states entangled with states of servers (possibly) outside of the colluding set. They could infer some more information about the file index by measuring these states, however, as this would consume the entanglement, the correctness of the protocol cannot be guaranteed anymore in this case. As the servers are honest-but-curious, \ie have to ensure that the user obtains the requested file, this is not possible with the specified number $\q{ent}$ of qubits sharing a maximally entangled state. On the other hand, any set of at most $t$ colluding nodes can in any case share all data between them, so increasing the number of qubits sharing a maximally entangled state within this colluding set does not breach privacy. It remains to be shown that the queries received by a set of at most $t$ colluding nodes reveal no information about the file index, which directly follows from the privacy of the protocol in \cite{tajeddine2018private}. For completeness, we include a short proof here for both user and server privacy. Consider a set $\cT \subset [n]$ with $|\cT|\leq t$ of colluding servers. The set of queries these servers receive is given by $\{ Q_s^{(p)}  \ | \ s\in \mathcal{T}, p\in [k]\}$. By the MDS property of the dual code $\Cs^\perp$ any subset of $t$ columns of $\GCp$ is linearly independent. As the vectors $Z_1^{(p)},...,Z_{n-k}^{(p)}$ are uniformly distributed and chosen independently for each $p \in [k]$, any subset of $t$ columns of
\begin{align*}
    \bbmatrix{Z^{(p)}_1 &\cdots & Z^{(p)}_{n-k}} \cdot \GCp
\end{align*}
is statistically independent and uniformly distributed. The sum of a uniformly distributed vector and an independently chosen vector is again uniformly distributed, and therefore adding the matrix $\bm{\xi}_{K,p}$ does not incur any dependence between any subset of $t$ columns and the file index $K$. Thus, user privacy is achieved. For each $l \in \qpar{L}$, server secrecy is achieved because in every round $b$ the received state of the user is independent of the fragments $y_{b,p}^{i,(l)}$ with $i \neq K$ and the measurement outcomes $G_s^{(l,b,p)}$ are mutually independent and independent of any file for all $s \in \gpar{2,\ldots,n-1}$.
\end{proof}

Unlike in the classical setting, the servers in the quantum setting do not need access to some shared randomness that is hidden from the user to achieve server secrecy. However, this should not be viewed as an inherent advantage since the servers instead share entanglement.
\begin{theorem} \label{thm:rate}
The QPIR rate of the scheme in Section~\ref{sec:scheme} is
\begin{align*}
    R_{\mathsf{QPIR}} &= \left\{ \begin{array}{ll} \frac{2}{k+t}, & \text{ if } k+t \text{ is even,}\\ \frac{2}{k+t+1}, & \text{ if } k+t \text{ is odd.} \  \end{array} \right. \ 
\end{align*}
\end{theorem}
\begin{proof}
As increasing $n$ will not affect the rate, let us assume $n = k+t$. See Remark \ref{rem:quantumLimitation} for more detailed discussion.  The upload cost is $U_\mathsf{QPIR} = k n \abs{Q} = k m n$ and $\lim_{F \to \infty} \frac{U_\mathsf{QPIR}}{F} = \lim_{\b \to \infty} \frac{k m n}{2k L \b} = 0$. Hence, we consider the upload cost to be negligible compared to the download cost for a large file size $F$. The download cost is $L n$ or $L (n+1)$ qubits per round and piece, \ie $\q{out} = k L n \b$ or $\q{out} = k L (n+1) \b$ qubits, if $n$ is even or odd, respectively. The information retrieved is a symbol of $\F_{4^L}$ per round and piece, \ie $I_\mathsf{QPIR} = k \b \log_2(4^L) = 2 k L \b$ bits. Thus, the rate is $R_\mathsf{QPIR} = \frac{2 k L \b}{k L n \b} = \frac{2}{n}$ if $n$ is even and $R_\mathsf{QPIR} = \frac{2 k L \b}{k L (n+1) \b} = \frac{2}{n+1}$ if $n$ is odd.
\end{proof}

\section{[4,2]-coded QPIR example}\label{sec:4-2RS}
Let us consider $n=4$ servers. Then $L=1$, and the base field is $\F_4 = \gpar{0,1,\a,\a^2}$ where $\a$ is a primitive element that satisfies
\begin{equation}
\a^2 + \a + 1 = 0.
\label{eq:F4 rule}
\end{equation}

\noindent  Suppose also $\b=1$, hence the set of files $x^i = \tpar{x_1^i,x_2^i}$ is given by $\mathcal{X} = \bigl\{x^i \in \tpar{\F_4}^2 : i \in \qpar{m} \bigr\}$. The files are encoded with a $\qpar{4,2}$-Reed--Solomon storage code with generator matrix
\[
\GC = \bbmatrix{1 & 0 & \a^2 & \a  \\
				 0 & 1 & \a   & \a^2}.
\]

\noindent Hence, the codewords are
\[
y^i = x^i \GC = \bbmatrix{x^i_1 & x^i_2 & \a^2 x^i_1 + \a x^i_2 & \a x^i_1 + \a^2 x^i_2}.
\]

\noindent Server $s \in \qpar{4}$ stores the $s$-th symbol of $y^i$, namely $y^i_s$.

\noindent \textbf{Preparation Step}. For each $p \in \qpar{2}$ the servers prepare the following qubits and states. Server $s \in \gpar{2,3}$ possesses 3 qubits $\H_s^{\Ls,(p)}$, $\H_s^{\Rs,(p)}$ and $\H_s^{(p)}$. The first and the last server possess qubits $\H_1^{(p)}$ and $\H_3^{(p)}$, respectively. Each pair $\tpar{\H_1^{(p)},\H_2^{\Ls,(p)}}$, $\tpar{\H_2^{\Rs,(p)},\H_3^{\Ls,(p)}}$, $\tpar{\H_3^{\Rs,(p)},\H_4^{(p)}}$ and $\tpar{\H_2^{(p)},\H_3^{(p)}}$ shares a maximally entangled state $\cPhi^{(p)}$. Thus, in this step servers prepare a total of $\q{in} = 2 \cdot (3 \cdot 2 + 2) = 16$ qubits and $\q{ent} = 2 \cdot 4 = 8$ pairs which share a maximally entangled state.

The protocol for querying the $K$-th file $x^K$ is depicted in \autoref{fig:QPIR example with RS-code} and is described as follows:
\begin{arabiclist}
	\item The user wants to retrieve the symbol $y^K_p$ stored in server $p \in \qpar{2}$. He generates two independent and uniformly random vectors $Z^{(p)}_1,Z^{(p)}_2 \in \tpar{\F_4}^m$, and encodes them as codewords of the dual code of $\Cs$. In other words, the user builds the queries $Q^{(p)}_1,\ldots,Q^{(p)}_4$ multiplying the random vectors by the generator matrix\footnote{The code chosen in this example is self-dual.} $\GCp = \GC$, \ie during round $p$ the queries are
	\[ \begin{split}
	& \bbmatrix{\! Q_1^{(p)} \! \! & Q_2^{(p)} \! \! & Q_3^{(p)} \! \! & Q_4^{(p)}} = \bbmatrix{Z_1^{(p)} \! \! & Z_2^{(p)} \!} \GCp \! + \bm{\xi}_{K,p} \\
	& = \bbmatrix{\! Z_1^{(p)} \! \! & Z_2^{(p)} \! \! & \a^2 Z_1^{(p)} \! + \! \a Z_2^{(p)} \! \! & \a Z_1^{(p)} \! + \! \a^2 Z_2^{(p)} \!} \! \! + \bm{\xi}_{K,p},
	\end{split} \]
	where $\bm{\xi}_{K,p}$ is a matrix whose $\tpar{K,p}$-th entry is 1 and all the other entries are 0.
	\item The user sends query $Q_s^{(p)}$ to each server $s\in \qpar{4}$.
	\item Server $s \in \qpar{4}$ computes $H_s^{(p)} = \scalar{Q^{(p)}_s}{y_s} \in \F_4$. For bijection \eqref{eq:Field isomorphism} each $H_s^{(p)}$ can be written as an element of $\F_2^2$. Then, the first and the last server apply $\W{H_1^{(p)}}$ and $\W{H_4^{(p)}}$ to the qubits $\H_1^{(p)}$ and $\H_4^{(p)}$, respectively. Server $s \in \gpar{2,3}$ applies $\W{H_s^{(p)}}$ to the qubit $\H_s^{\Ls,(p)}$ and performs a Bell measurement on $\H_s^{\Ls,(p)} \otimes \H_s^{\Rs,(p)}$ whose outcome is denoted by $G_s^{(p)} \in \F_2^2$. Finally, server $s$ applies $\W{G_s^{(p)}}$ to the qubit $\H_s^{(p)}$. This last operation initializes the two-sum transmission protocol (cf. Sec.~\ref{sec:QC}) between server 2, server 3 and the user: the input qubits $\H_2^{(p)} \otimes \H_3^{(p)}$ are used to send $G_2^{(p)} + G_3^{(p)}$ to the user.
	\item Each server sends its qubit $\H_s^{(p)}$ to the user.
	\item The user performs a Bell measurement to the pair $\H_2^{(p)} \otimes \H_3^{(p)}$ to retrieve $G^{(p)} = G_2^{(p)} + G_3^{(p)}$ via the two-sum transmission protocol. He applies $\W{G^{(p)}}$ to the qubit $\H_4^{(p)}$ and performs a Bell measurement on $\H_1^{(p)} \otimes \H_4^{(p)}$, whose outcome is $y^K_p$ with probability 1.
	\item Repeat all the previous steps for every piece $p \in \qpar{2}$.
	\item Since servers 1 and 2 contain the symbols $x^i_1$ and $x^i_2$ respectively, the user can directly build the file $x^K = \tpar{x^K_1,x^K_2}$.
\end{arabiclist}
\begin{figure}[h]
\centering
\includegraphics[width = 0.5\textwidth]{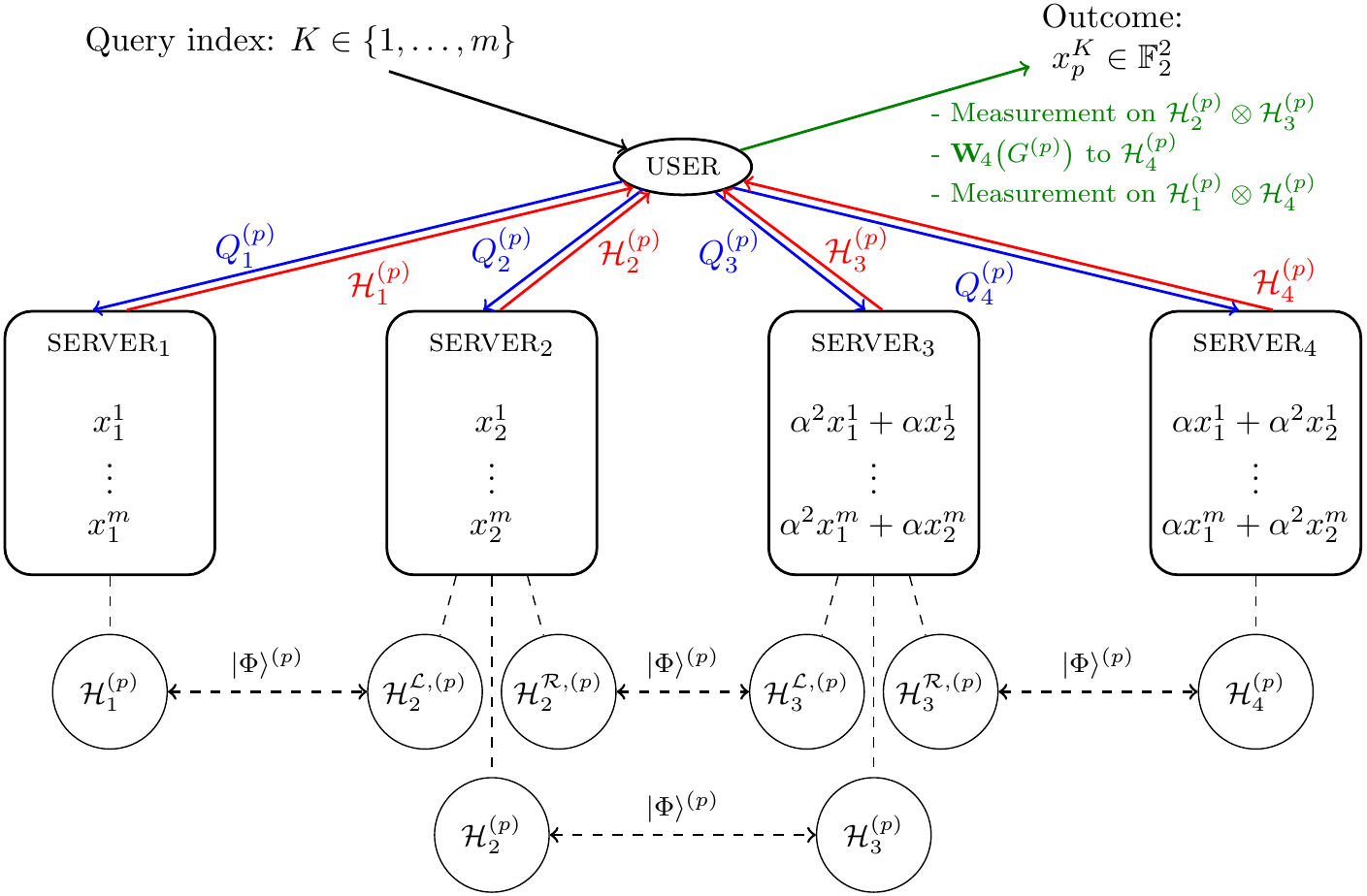}
\caption{QPIR protocol with an $\qpar{4,2}$-RS code, as described in Section~\ref{sec:4-2RS}. Step 2 is depicted in blue, Step 4 in red and Step 5 in green.} 
\label{fig:QPIR example with RS-code}
\end{figure}
\begin{figure}[h]
\centering
\includegraphics[width = .4\textwidth]{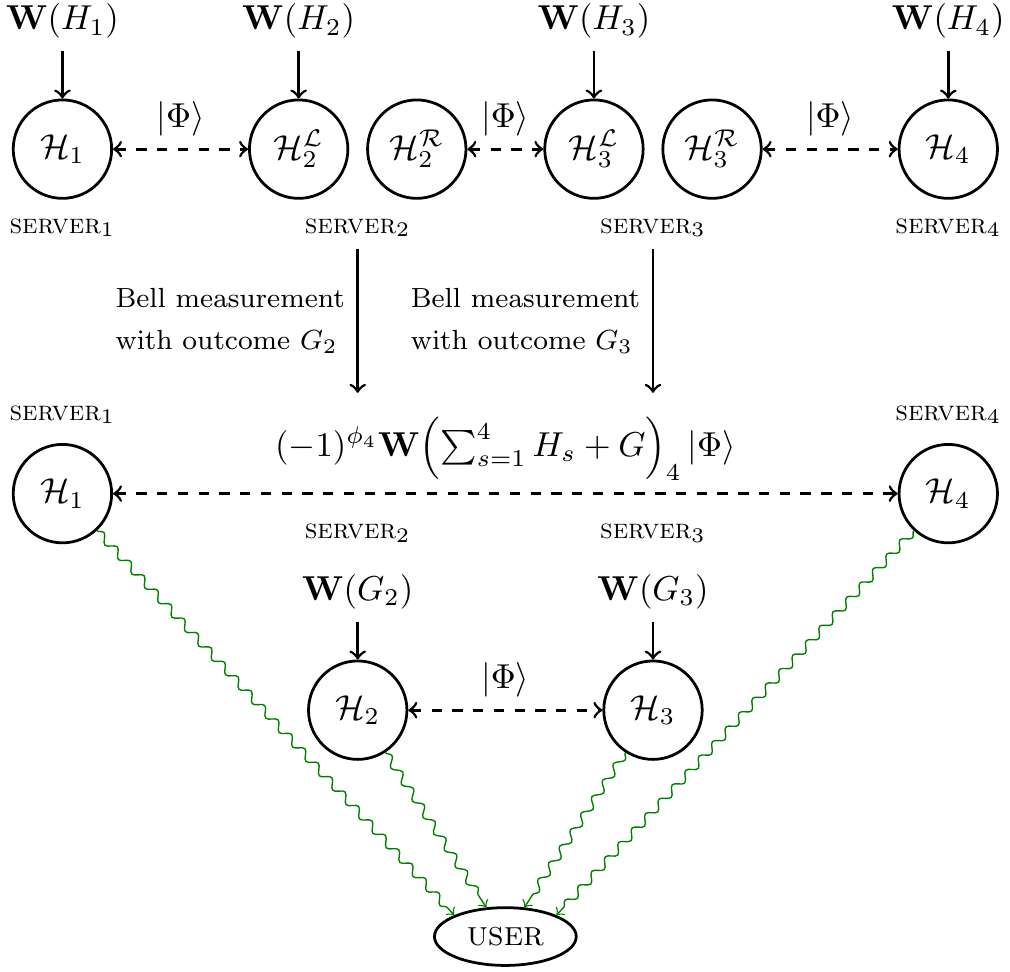}
\caption{Servers' operations and download step (depicted in green) for the QPIR protocol with a $\qpar{4,2}$-RS code, as described in Section~\ref{sec:4-2RS}. These operations are performed for every $p \in \qpar{2}$. Here, we dropped the superscript $(p)$ for a clearer visualization.} 
\label{fig:QPIR download with RS-code}
\end{figure}

The operations performed by the servers and the download step (Step 4) are visualized in \autoref{fig:QPIR download with RS-code}. After these steps, the user possesses the entangled pairs of qubits $\H_2^{(p)} \otimes \H_3^{(p)}$ and $\H_1^{(p)} \otimes \H_4^{(p)}$ for each $p \in \qpar{2}$. The states of those two pairs are, respectively,
\begin{gather*}
\Wi{2}{G_2^{(p)}} \Wi{3}{G_3^{(p)}} \cPhi^{(p)} \smalloverbrace{=}^{\hspace{-4pt} \eqref{eq:Weyl system moving}, \hspace{1pt} \eqref{eq:Weyl sum} \hspace{-4pt}} (-1)^{\phi_G^{(p)}} \Wi{3}{G^{(p)}} \cPhi^{(p)}, \\
(-1)^{\phi_4^{(p)}} \Wi{4}{\sum_{s=1}^4 H_s^{(p)} + G^{(p)}} \cPhi^{(p)},
\end{gather*}

\noindent where $\sum_{s=1}^4 H_s^{(p)} = x^K_p$, $\phi_G^{(p)} \in \F_2$ is determined upon $G_2^{(p)}$ and $G_3^{(p)}$, and $\phi_4^{(p)} \in \F_2$ is determined upon $H_1^{(p)}, H_2^{(p)}, H_3^{(p)}, H_4^{(p)}$ and $G^{(p)}$. In fact, assuming $p=1$,
\[ \begin{split}
\sum_{s=1}^4 H_s^{(1)} & = \sum_{i=1}^m \tpar{Q^i_1 + \delta_{iK}} y^i_1 + \sum_{i=1}^m Q^i_2 y^i_2 + \sum_{i=1}^m Q^i_3 y^i_3 \\
& + \sum_{i=1}^m Q^i_4 y^i_4  = \sum_{i=1}^m Z^i_1 x^i_1 + x^K_1 + \sum_{i=1}^m Z^i_2 x^i_2 \\
& + \sum_{i=1}^m \tpar{\a^2 Z^i_1 + \a Z^i_2} \tpar{\a^2 x^i_1 + \a x^i_2} \\
& + \sum_{i=1}^m \tpar{\a Z^i_1 + \a^2 Z^i_2} \tpar{\a x^i_1 + \a^2 x^i_2} \\
& \smalloverbrace{=}^{\eqref{eq:F4 rule}} x_1^K + \sum_{i=1}^m \textcolor{blue}{Z^i_1 x^i_1} + \sum_{i=1}^m \textcolor{cyan}{Z^i_2 x^i_2} + \sum_{i=1}^m \textcolor{blue}{\a Z^i_1 x^i_1} \\
& + \! \sum_{i=1}^m \textcolor{red}{Z^i_1 x^i_2} \! + \! \sum_{i=1}^m \textcolor{magenta}{Z^i_2 x^i_1} \! + \! \sum_{i=1}^m \textcolor{cyan}{\a^2 Z^i_2 x^i_2} \! + \! \sum_{i=1}^m \textcolor{blue}{\a^2 Z^i_1 x^i_1} \\
& + \sum_{i=1}^m \textcolor{red}{Z^i_1 x^i_2} + \sum_{i=1}^m \textcolor{magenta}{Z^i_2 x^i_1} + \sum_{i=1}^m \textcolor{cyan}{\a Z^i_2 x^i_2} \smalloverbrace{=}^{\eqref{eq:F4 rule}} x^K_1.
\end{split} \]

\noindent The proof for $p=2$ is the same. Performing a Bell measurement on the first pair, the user retrieves $G^{(p)}$ and then applies $\W{G^{(p)}}$ to the qubit $\H_4^{(p)}$. Doing so, the state of the second pair becomes
\[ \begin{split}
    & (-1)^{\phi_4^{(p)}} \Wi{4}{G^{(p)}} \Wi{4}{x^K_p + G^{(p)}} \cPhi^{(p)} \\
    & \stackrel{\eqref{eq:Weyl sum}}{=} (-1)^{\phi_u^{(p)}} \Wi{4}{x^K_p} \cPhi^{(p)},
\end{split} \]

\noindent where $\phi_u^{(p)} \in \F_2$ is determined upon $\phi_4^{(p)}, G^{(p)}$ and $x_p^K + G^{(p)}$. Performing now the Bell measurement on the second pair, the user retrieves $x^K_p$ with probability 1. After retrieving $x^K_1$ and $x^K_2$, the user depacketizes the requested file $x^K$.

User and server secrecy follow by Lemma~\ref{lem:Secrecy}. We now notice that a total of 4 qubits are transmitted to the user in each instance of Step 4. Hence, the rate is $R_\mathsf{QPIR} = \frac{2 \cdot 2}{2 \cdot 4} = \frac{1}{2}$, since we recover four bits and download $\q{out} = 8$ qubits over two rounds.

For the related quantum circuits, see the Figure in Appendix.

\section{QPIR from databases encoded with locally repairable codes} \label{sec:LRC}

In this section we consider the application of the introduced protocol to a DSS encoded with an LRC  \cite{huang2013pyramid,gopalan2012locality, kamath2014codes}, a code that allows for a small number code servers to be recovered locally, \emph{i.e.}, by only downloading data from a small subset of other servers that are in the same \emph{repair set}\footnote{Note that in order to retrieve the \emph{message} (part of a file) from a (non-trivial) LRC codeword, it is not sufficient to access a single repair set, \emph{i.e.}, any information set of an LRC consists of positions from multiple repair sets.}. For the classical setting this has been treated in \cite{martinez2019private, kumar2017private,freij2017private}. Locality comes at the cost of a smaller minimum distance of the code, thereby establishing a trade-off between the locality and maximum number of server failures that is guaranteed to not induce data loss. The restriction of the rate for the scheme presented in Section~\ref{sec:scheme} arises from only being able to download a single symbol of the codeword corresponding to the desired file in each round (cf. Rem.~\ref{rem:quantumLimitation}). Interestingly, for a scheme with this restriction, regardless of whether it is a quantum or a classical system, it can be advantageous in terms of the retrieval rate to operate on a storage system that is encoded with a locally repairable code (LRC). Without such a restriction on the number of downloaded symbols, LRCs are not known to be able to improve the retrieval rates. Here, we only consider LRCs with disjoint repair sets, as formally defined in Definition~\ref{def:LRC}. For a set of integers $\A \subset [n]$ denote the restriction of the $[n,k]$ code $\mathcal{C}$ to the coordinates indexed by $\A$ by $\mathcal{C}|_{\A}$.
\begin{definition}\label{def:LRC}
An $[n,k]$ code $\Cs$ is said to have $(r,\rho)$-locality if there exists a partition $\mathcal{P} = \{\A_1,...,\A_{\mu}\}$ of $[n]$ into sets $\A_l$ with $\left|\A_l\right| \leq r+\rho-1, \ \forall l \in [\mu]$ such that for the distance of the code restricted to the positions indexed by $\A_l$ it holds that $d(\Cs_{\A_l}) \geq \rho , \ \forall \ l\in [\mu]$.
\end{definition}
The sets $\A_1,...,\A_\mu$ are referred to as the repair sets of the code and the restriction $\Cs|_{\A_l}$ of the code $\Cs$ to the positions of $\A_l$ as the $l$-th local code. 
We say a locally repairable code  is \emph{optimal} if it fulfills the Singleton-like bound \cite{gopalan2012locality,kamath2014codes}
\begin{equation*}
  d(\Cs)\leq n-k+1-\left(\left\lceil \frac{k}{r}\right\rceil -1\right)(\rho-1)
\end{equation*}
with equality. A multitude of such optimal code families exist, e.g., see \cite{huang2013pyramid,kamath2014codes,tamo2014family}. For simplicity, here we assume that $r|k$, so we can omit the ceiling operation. For an optimal LRC the local codes are $[r+\rho-1,r]$-MDS codes \cite{huang2013pyramid,kamath2014codes}. It can readily be seen that picking from $\frac{k}{r}$ local codes arbitrary $r$ positions each results in an information set, by restricting the LRC to a subset $\mathcal{S} \subset [n]$ of positions that is the union of $\frac{k}{r}$ repair sets. Then the distance is
\[ \begin{split}
    d(\Cs|_\mathcal{S}) & \geq n-k+1-\left(\frac{k}{r} -1\right)(\rho-1) \\
    & - \underbrace{\left( \frac{n}{r+\rho-1} - \frac{k}{r} \right) (r+\rho-1)}_{\text{maximal decrease in distance due to puncturing}} = \rho
\end{split} \]
and since each of the remaining repair sets is an $[r+\rho-1,r]$ MDS code, it holds that $d(\Cs|_{\mathcal{S}})=\rho$. Clearly, puncturing each of these local codes in any $\rho-1$ positions does not decrease the dimension of the local code, and thereby also not the overall dimension. It follows that the remaining $k$ positions form an information set of the LRC.

In the following we denote by $\mathsf{PIR}(n,k,t,i,j)$ a PIR scheme that retrieves $1$ symbol with index $j$ of the codeword corresponding to the $i$-th file from an $[n,k]$ MDS coded storage system. The scheme for DSSs encoded with an LRC presented here relies on the repeated application of such a scheme to specific subsets of positions/servers. Note that the scheme introduced in Section~\ref{sec:scheme} can be viewed as a $\mathsf{PIR}(n,k,n-k,i,j)$ scheme applied in multiple rounds, once for each $j\in [k]$.

Now let the storage code $\mathcal{C}$ be an optimal $[n,k,r,\rho]$ LRC as in Definition~\ref{def:LRC} and $\mathcal{I} \subset \bigcup_{l=1}^{k/r} A_l$ be such that $|\mathcal{I} \cap A_l| = r , \ \forall l\in[\frac{k}{r}]$ (w.l.o.g. we choose the first $\frac{k}{r}$ repair sets $\A_1,...,\A_{\frac{k}{r}}$). As shown above, performing the PIR scheme $\mathsf{PIR}(r+\rho-1, r, \rho-1, i, j), \ \forall j \in \mathcal{I}$ on the MDS code $\Cs_{\A_l}$ given by the repair set with $j \in \A_l$, results in an information set of the LRC and therefore allows for the $i$-th file to be recovered.

\begin{corollary}\label{C-cor}
 Let $\mathsf{PIR}(r+\rho-1, r, \rho-1, i, j)$ be the QPIR scheme of Section~\ref{sec:scheme}. Then, the retrieval rate of the above QPIR scheme from a DSS encoded with an $[n,k,r,\rho]$ LRC is
\begin{equation*}
    R_{\mathsf{QPIR}} =  \left\{ \begin{array}{ll} \frac{2}{r+t}, & \text{ if } r+t \text{ is even,}\\ \frac{2}{r+t+1}, & \text{ if } r+t \text{ is odd,} \  \end{array} \right. \ .
\end{equation*}
for $0\leq t<\rho$.
\end{corollary}

\begin{proof}
By Theorem~\ref{thm:rate} the rate of the scheme from Section~\ref{sec:scheme}, which is applied repeatedly to retrieve all symbols indexed by $\mathcal{I}$, is given by
\begin{equation*}
  R_{\mathsf{QPIR}}^{\mathsf{local}} = \left\{ \begin{array}{ll} \frac{2}{r+t}, & \text{ if } r+t \text{ is even,}\\ \frac{2}{r+t+1}, & \text{ if } r+t \text{ is odd,} \  \end{array} \right. \ 
\end{equation*}
where $t= \rho-1$. As the scheme is applied to the local repair sets a total of $k$ times to obtain $k$ symbols of the desired file, the overall rate is equal to the local rate $R_{\mathsf{QPIR}}^{\mathsf{local}}$.
\end{proof}

Note that for LRC coded storage the QPIR rate does no longer directly depend on the length $n$ of the code. Since the locality parameter $r$ is usually considerably smaller than the code dimension $k$, this results in an increase in the retrieval rate compared to the rate for MDS coded storage of Theorem~\ref{thm:rate}. On the other hand, as $\rho$ is typically small, this comes at the cost of a lowered resistance against collusion and total number of server failures that can be tolerated. However, the scheme is still able to resist more than $t=\rho-1$ colluding servers, provided that no more than $t$ servers collude per local group. Explicitly, for such collusion patterns, the scheme can resist collusion of up to $t\mu=(\rho-1)\mu$ servers.

\section{Conclusions and future work}

In the classical private information retrieval (PIR) setup, a user wants to retrieve a file from a database or a distributed storage system (DSS)  without revealing the file identity to the servers holding the data. In the quantum PIR (QPIR) setting, a user privately retrieves a classical file by receiving quantum information from the servers. The QPIR problem has been treated by Song \emph{et al.} in the case of replicated servers, both without collusion and with all but one servers colluding. In this paper, the QPIR setting was extended to  account for maximum distance separable (MDS) coded servers. The proposed protocol works for any $[n,k]$-MDS code and $t$-collusion with $t=n-k$. Similarly to the previous cases, the rates achieved are better than those known or conjectured in the classical counterparts. It was also demonstrated how the retrieval rates can be significantly improved by using locally repairable codes (LRCs) consisting of disjoint repair groups, each of which is an MDS code. Our next goal is to find the capacity of QPIR from coded and colluding servers. We may also consider different approaches (\emph{e.g.}, divide the servers into pairs and use the two-sum transmission protocol from each pair to the user), malicious servers and arbitrary coded servers in future work.

\section*{Acknowledgments}
The authors would like to thank Prof.~M. Hayashi and S.~Song for helpful discussions.

\bibliographystyle{IEEEtran}
\bibliography{main}

\newpage \ \newpage

\twocolumn[
  \begin{@twocolumnfalse}
   \section*{Appendix: Quantum circuit for the QPIR example of Sec. \ref{sec:4-2RS}. }
   \vspace{1cm}

\centering
\begin{tikzpicture}[scale=1.50000,x=1pt,y=1pt]
\filldraw[color=white] (17.500000, 3.000000) rectangle (-268.500000, -408.000000);
\usetikzlibrary{decorations.pathreplacing,decorations.pathmorphing}
\draw[color=black] (-251.000000,10.000000) node[above] {$\textsc{server}_1$};
\draw[color=black] (-179.000000,10.000000) node[above] {$\textsc{server}_2$};
\draw[color=black] (-72.000000,10.000000) node[above] {$\textsc{server}_3$};
\draw[color=black] (-0.000000,10.000000) node[above] {$\textsc{server}_4$};
\draw[color=black,thick] (-251.000000,0.000000) -- (-251.000000,-369.000000);
\draw[color=black,thick] (-250.500000,-369.000000) -- (-250.500000,-394.500000);
\draw[color=black,thick] (-251.500000,-369.000000) -- (-251.500000,-394.500000);
\draw[color=black] (-251.000000,0.000000) node[above] {$\mathcal{H}_1$};
\draw[color=black,dashed] (-232.500000,18.000000) -- (-232.500000,-211.500000);
\draw[color=black,thick] (-214.000000,0.000000) -- (-214.000000,-106.500000);
\draw[color=black,thick] (-213.500000,-106.500000) -- (-213.500000,-186.000000);
\draw[color=black,thick] (-214.500000,-106.500000) -- (-214.500000,-186.000000);
\draw[color=black] (-214.000000,0.000000) node[above] {$\mathcal{H}_2^\mathcal{L}$};
\draw[color=black,thick] (-179.000000,-133.500000) -- (-179.000000,-279.000000);
\draw[color=black,thick] (-178.500000,-279.000000) -- (-178.500000,-303.000000);
\draw[color=black,thick] (-179.500000,-279.000000) -- (-179.500000,-303.000000);
\draw[color=black] (-179.000000,0.000000) node[above] {${}$};
\draw[color=black,thick] (-144.000000,0.000000) -- (-144.000000,-106.500000);
\draw[color=black,thick] (-143.500000,-106.500000) -- (-143.500000,-186.000000);
\draw[color=black,thick] (-144.500000,-106.500000) -- (-144.500000,-186.000000);
\draw[color=black] (-144.000000,0.000000) node[above] {$\mathcal{H}_2^\mathcal{R}$};
\draw[color=black,dashed] (-125.500000,18.000000) -- (-125.500000,-211.500000);
\draw[color=black,thick] (-107.000000,0.000000) -- (-107.000000,-106.500000);
\draw[color=black,thick] (-106.500000,-106.500000) -- (-106.500000,-186.000000);
\draw[color=black,thick] (-107.500000,-106.500000) -- (-107.500000,-186.000000);
\draw[color=black] (-107.000000,0.000000) node[above] {$\mathcal{H}_3^\mathcal{L}$};
\draw[color=black,thick] (-72.000000,-133.500000) -- (-72.000000,-279.000000);
\draw[color=black,thick] (-71.500000,-279.000000) -- (-71.500000,-303.000000);
\draw[color=black,thick] (-72.500000,-279.000000) -- (-72.500000,-303.000000);
\draw[color=black] (-72.000000,0.000000) node[above] {${}$};
\draw[color=black,thick] (-37.000000,0.000000) -- (-37.000000,-106.500000);
\draw[color=black,thick] (-36.500000,-106.500000) -- (-36.500000,-186.000000);
\draw[color=black,thick] (-37.500000,-106.500000) -- (-37.500000,-186.000000);
\draw[color=black] (-37.000000,0.000000) node[above] {$\mathcal{H}_3^\mathcal{R}$};
\draw[color=black,dashed] (-18.500000,18.000000) -- (-18.500000,-211.500000);
\draw[color=black,thick] (-0.000000,0.000000) -- (-0.000000,-369.000000);
\draw[color=black,thick] (0.500000,-369.000000) -- (0.500000,-394.500000);
\draw[color=black,thick] (-0.500000,-369.000000) -- (-0.500000,-394.500000);
\draw[color=black] (-0.000000,0.000000) node[above] {$\mathcal{H}_4$};
\draw[fill=white,color=white] (-257.000000, -21.000000) rectangle (-208.000000, -6.000000);
\draw (-232.500000, -13.500000) node {$\cPhi$};
\draw[decorate,decoration={brace,amplitude = 4.000000pt},very thick] (-208.000000,-6.000000) -- (-257.000000,-6.000000);
\draw[fill=white,color=white] (-150.000000, -21.000000) rectangle (-101.000000, -6.000000);
\draw (-125.500000, -13.500000) node {$\cPhi$};
\draw[decorate,decoration={brace,amplitude = 4.000000pt},very thick] (-101.000000,-6.000000) -- (-150.000000,-6.000000);
\draw[fill=white,color=white] (-43.000000, -21.000000) rectangle (6.000000, -6.000000);
\draw (-18.500000, -13.500000) node {$\cPhi$};
\draw[decorate,decoration={brace,amplitude = 4.000000pt},very thick] (6.000000,-6.000000) -- (-43.000000,-6.000000);
\begin{scope}
\draw[fill=white] (-251.000000, -39.000000) +(-45.000000:21.213203pt and 8.485281pt) -- +(45.000000:21.213203pt and 8.485281pt) -- +(135.000000:21.213203pt and 8.485281pt) -- +(225.000000:21.213203pt and 8.485281pt) -- cycle;
\clip (-251.000000, -39.000000) +(-45.000000:21.213203pt and 8.485281pt) -- +(45.000000:21.213203pt and 8.485281pt) -- +(135.000000:21.213203pt and 8.485281pt) -- +(225.000000:21.213203pt and 8.485281pt) -- cycle;
\draw (-251.000000, -39.000000) node {$\mathbf{W}_{1}\!\left( H_1 \right)$};
\end{scope}
\begin{scope}
\draw[fill=white] (-214.000000, -39.000000) +(-45.000000:24.748737pt and 8.485281pt) -- +(45.000000:24.748737pt and 8.485281pt) -- +(135.000000:24.748737pt and 8.485281pt) -- +(225.000000:24.748737pt and 8.485281pt) -- cycle;
\clip (-214.000000, -39.000000) +(-45.000000:24.748737pt and 8.485281pt) -- +(45.000000:24.748737pt and 8.485281pt) -- +(135.000000:24.748737pt and 8.485281pt) -- +(225.000000:24.748737pt and 8.485281pt) -- cycle;
\draw (-214.000000, -39.000000) node {$\mathbf{W}_{2^\mathcal{L}}\!\left( H_2 \right)$};
\end{scope}
\begin{scope}
\draw[fill=white] (-107.000000, -39.000000) +(-45.000000:24.748737pt and 8.485281pt) -- +(45.000000:24.748737pt and 8.485281pt) -- +(135.000000:24.748737pt and 8.485281pt) -- +(225.000000:24.748737pt and 8.485281pt) -- cycle;
\clip (-107.000000, -39.000000) +(-45.000000:24.748737pt and 8.485281pt) -- +(45.000000:24.748737pt and 8.485281pt) -- +(135.000000:24.748737pt and 8.485281pt) -- +(225.000000:24.748737pt and 8.485281pt) -- cycle;
\draw (-107.000000, -39.000000) node {$\mathbf{W}_{3^\mathcal{L}}\!\left( H_3 \right)$};
\end{scope}
\begin{scope}
\draw[fill=white] (0.000000, -39.000000) +(-45.000000:21.213203pt and 8.485281pt) -- +(45.000000:21.213203pt and 8.485281pt) -- +(135.000000:21.213203pt and 8.485281pt) -- +(225.000000:21.213203pt and 8.485281pt) -- cycle;
\clip (0.000000, -39.000000) +(-45.000000:21.213203pt and 8.485281pt) -- +(45.000000:21.213203pt and 8.485281pt) -- +(135.000000:21.213203pt and 8.485281pt) -- +(225.000000:21.213203pt and 8.485281pt) -- cycle;
\draw (0.000000, -39.000000) node {$\mathbf{W}_{4}\!\left( H_4 \right)$};
\end{scope}
\draw (-214.000000,-60.000000) -- (-144.000000,-60.000000);
\filldraw (-214.000000, -60.000000) circle(1.500000pt);
\begin{scope}
\draw[fill=white] (-144.000000, -60.000000) circle(3.000000pt);
\clip (-144.000000, -60.000000) circle(3.000000pt);
\draw (-147.000000, -60.000000) -- (-141.000000, -60.000000);
\draw (-144.000000, -63.000000) -- (-144.000000, -57.000000);
\end{scope}
\draw (-107.000000,-60.000000) -- (-37.000000,-60.000000);
\filldraw (-107.000000, -60.000000) circle(1.500000pt);
\begin{scope}
\draw[fill=white] (-37.000000, -60.000000) circle(3.000000pt);
\clip (-37.000000, -60.000000) circle(3.000000pt);
\draw (-40.000000, -60.000000) -- (-34.000000, -60.000000);
\draw (-37.000000, -63.000000) -- (-37.000000, -57.000000);
\end{scope}
\begin{scope}
\draw[fill=white] (-214.000000, -81.000000) +(-45.000000:8.485281pt and 8.485281pt) -- +(45.000000:8.485281pt and 8.485281pt) -- +(135.000000:8.485281pt and 8.485281pt) -- +(225.000000:8.485281pt and 8.485281pt) -- cycle;
\clip (-214.000000, -81.000000) +(-45.000000:8.485281pt and 8.485281pt) -- +(45.000000:8.485281pt and 8.485281pt) -- +(135.000000:8.485281pt and 8.485281pt) -- +(225.000000:8.485281pt and 8.485281pt) -- cycle;
\draw (-214.000000, -81.000000) node {$H$};
\end{scope}
\begin{scope}
\draw[fill=white] (-107.000000, -81.000000) +(-45.000000:8.485281pt and 8.485281pt) -- +(45.000000:8.485281pt and 8.485281pt) -- +(135.000000:8.485281pt and 8.485281pt) -- +(225.000000:8.485281pt and 8.485281pt) -- cycle;
\clip (-107.000000, -81.000000) +(-45.000000:8.485281pt and 8.485281pt) -- +(45.000000:8.485281pt and 8.485281pt) -- +(135.000000:8.485281pt and 8.485281pt) -- +(225.000000:8.485281pt and 8.485281pt) -- cycle;
\draw (-107.000000, -81.000000) node {$H$};
\end{scope}
\draw[fill=white] (-220.000000, -112.500000) rectangle (-208.000000, -100.500000);
\draw[very thin] (-214.000000, -105.900000) arc (90:150:6.000000pt);
\draw[very thin] (-214.000000, -105.900000) arc (90:30:6.000000pt);
\draw[->,>=stealth] (-214.000000, -111.900000) -- +(80:10.392305pt);
\draw[fill=white] (-150.000000, -112.500000) rectangle (-138.000000, -100.500000);
\draw[very thin] (-144.000000, -105.900000) arc (90:150:6.000000pt);
\draw[very thin] (-144.000000, -105.900000) arc (90:30:6.000000pt);
\draw[->,>=stealth] (-144.000000, -111.900000) -- +(80:10.392305pt);
\draw[color=black] (-179.000000,-99.000000) node[fill=white,below,minimum width=35.000000pt,minimum height=15.000000pt,inner sep=0pt] {\phantom{${}$}};
\draw[color=black] (-179.000000,-99.000000) node[below] {${}$};
\draw[fill=white] (-113.000000, -112.500000) rectangle (-101.000000, -100.500000);
\draw[very thin] (-107.000000, -105.900000) arc (90:150:6.000000pt);
\draw[very thin] (-107.000000, -105.900000) arc (90:30:6.000000pt);
\draw[->,>=stealth] (-107.000000, -111.900000) -- +(80:10.392305pt);
\draw[fill=white] (-43.000000, -112.500000) rectangle (-31.000000, -100.500000);
\draw[very thin] (-37.000000, -105.900000) arc (90:150:6.000000pt);
\draw[very thin] (-37.000000, -105.900000) arc (90:30:6.000000pt);
\draw[->,>=stealth] (-37.000000, -111.900000) -- +(80:10.392305pt);
\draw[color=black] (-72.000000,-99.000000) node[fill=white,below,minimum width=35.000000pt,minimum height=15.000000pt,inner sep=0pt] {\phantom{${}$}};
\draw[color=black] (-72.000000,-99.000000) node[below] {${}$};
\draw[color=black] (-179.000000,-141.000000) node[fill=white,above,minimum width=35.000000pt,minimum height=15.000000pt,inner sep=0pt] {\phantom{$\mathcal{H}_2$}};
\draw[color=black] (-179.000000,-141.000000) node[above] {$\mathcal{H}_2$};
\draw[color=black] (-72.000000,-141.000000) node[fill=white,above,minimum width=35.000000pt,minimum height=15.000000pt,inner sep=0pt] {\phantom{$\mathcal{H}_3$}};
\draw[color=black] (-72.000000,-141.000000) node[above] {$\mathcal{H}_3$};
\draw[fill=white,color=white] (-185.000000, -168.000000) rectangle (-66.000000, -153.000000);
\draw (-125.500000, -160.500000) node {$\cPhi$};
\draw[decorate,decoration={brace,amplitude = 4.000000pt},very thick] (-66.000000,-153.000000) -- (-185.000000,-153.000000);
\draw (-214.000000,-185.500000) -- (-179.000000,-185.500000);
\draw (-214.000000,-186.500000) -- (-179.000000,-186.500000);
\draw (-179.000000,-185.500000) -- (-144.000000,-185.500000);
\draw (-179.000000,-186.500000) -- (-144.000000,-186.500000);
\begin{scope}
\draw[fill=white] (-179.000000, -186.000000) +(-45.000000:21.213203pt and 8.485281pt) -- +(45.000000:21.213203pt and 8.485281pt) -- +(135.000000:21.213203pt and 8.485281pt) -- +(225.000000:21.213203pt and 8.485281pt) -- cycle;
\clip (-179.000000, -186.000000) +(-45.000000:21.213203pt and 8.485281pt) -- +(45.000000:21.213203pt and 8.485281pt) -- +(135.000000:21.213203pt and 8.485281pt) -- +(225.000000:21.213203pt and 8.485281pt) -- cycle;
\draw (-179.000000, -186.000000) node {$\mathbf{W}_{2}\!\left( G_2 \right)$};
\end{scope}
\filldraw (-214.000000, -186.000000) circle(1.500000pt);
\filldraw (-144.000000, -186.000000) circle(1.500000pt);
\draw (-107.000000,-185.500000) -- (-72.000000,-185.500000);
\draw (-107.000000,-186.500000) -- (-72.000000,-186.500000);
\draw (-72.000000,-185.500000) -- (-37.000000,-185.500000);
\draw (-72.000000,-186.500000) -- (-37.000000,-186.500000);
\begin{scope}
\draw[fill=white] (-72.000000, -186.000000) +(-45.000000:21.213203pt and 8.485281pt) -- +(45.000000:21.213203pt and 8.485281pt) -- +(135.000000:21.213203pt and 8.485281pt) -- +(225.000000:21.213203pt and 8.485281pt) -- cycle;
\clip (-72.000000, -186.000000) +(-45.000000:21.213203pt and 8.485281pt) -- +(45.000000:21.213203pt and 8.485281pt) -- +(135.000000:21.213203pt and 8.485281pt) -- +(225.000000:21.213203pt and 8.485281pt) -- cycle;
\draw (-72.000000, -186.000000) node {$\mathbf{W}_{3}\!\left( G_3 \right)$};
\end{scope}
\filldraw (-107.000000, -186.000000) circle(1.500000pt);
\filldraw (-37.000000, -186.000000) circle(1.500000pt);
\draw (-179.000000,-234.000000) -- (-72.000000,-234.000000);
\filldraw (-179.000000, -234.000000) circle(1.500000pt);
\begin{scope}
\draw[fill=white] (-72.000000, -234.000000) circle(3.000000pt);
\clip (-72.000000, -234.000000) circle(3.000000pt);
\draw (-75.000000, -234.000000) -- (-69.000000, -234.000000);
\draw (-72.000000, -237.000000) -- (-72.000000, -231.000000);
\end{scope}
\begin{scope}
\draw[fill=white] (-179.000000, -255.000000) +(-45.000000:8.485281pt and 8.485281pt) -- +(45.000000:8.485281pt and 8.485281pt) -- +(135.000000:8.485281pt and 8.485281pt) -- +(225.000000:8.485281pt and 8.485281pt) -- cycle;
\clip (-179.000000, -255.000000) +(-45.000000:8.485281pt and 8.485281pt) -- +(45.000000:8.485281pt and 8.485281pt) -- +(135.000000:8.485281pt and 8.485281pt) -- +(225.000000:8.485281pt and 8.485281pt) -- cycle;
\draw (-179.000000, -255.000000) node {$H$};
\end{scope}
\draw[fill=white] (-185.000000, -285.000000) rectangle (-173.000000, -273.000000);
\draw[very thin] (-179.000000, -278.400000) arc (90:150:6.000000pt);
\draw[very thin] (-179.000000, -278.400000) arc (90:30:6.000000pt);
\draw[->,>=stealth] (-179.000000, -284.400000) -- +(80:10.392305pt);
\draw[fill=white] (-78.000000, -285.000000) rectangle (-66.000000, -273.000000);
\draw[very thin] (-72.000000, -278.400000) arc (90:150:6.000000pt);
\draw[very thin] (-72.000000, -278.400000) arc (90:30:6.000000pt);
\draw[->,>=stealth] (-72.000000, -284.400000) -- +(80:10.392305pt);
\draw (-179.000000,-302.500000) -- (-0.000000,-302.500000);
\draw (-179.000000,-303.500000) -- (-0.000000,-303.500000);
\begin{scope}
\draw[fill=white] (0.000000, -303.000000) +(-45.000000:21.213203pt and 8.485281pt) -- +(45.000000:21.213203pt and 8.485281pt) -- +(135.000000:21.213203pt and 8.485281pt) -- +(225.000000:21.213203pt and 8.485281pt) -- cycle;
\clip (0.000000, -303.000000) +(-45.000000:21.213203pt and 8.485281pt) -- +(45.000000:21.213203pt and 8.485281pt) -- +(135.000000:21.213203pt and 8.485281pt) -- +(225.000000:21.213203pt and 8.485281pt) -- cycle;
\draw (0.000000, -303.000000) node {$\mathbf{W}_{4}\!\left( G \right)$};
\end{scope}
\filldraw (-179.000000, -303.000000) circle(1.500000pt);
\filldraw (-72.000000, -303.000000) circle(1.500000pt);
\draw (-251.000000,-324.000000) -- (-0.000000,-324.000000);
\filldraw (-251.000000, -324.000000) circle(1.500000pt);
\begin{scope}
\draw[fill=white] (-0.000000, -324.000000) circle(3.000000pt);
\clip (-0.000000, -324.000000) circle(3.000000pt);
\draw (-3.000000, -324.000000) -- (3.000000, -324.000000);
\draw (-0.000000, -327.000000) -- (-0.000000, -321.000000);
\end{scope}
\begin{scope}
\draw[fill=white] (-251.000000, -345.000000) +(-45.000000:8.485281pt and 8.485281pt) -- +(45.000000:8.485281pt and 8.485281pt) -- +(135.000000:8.485281pt and 8.485281pt) -- +(225.000000:8.485281pt and 8.485281pt) -- cycle;
\clip (-251.000000, -345.000000) +(-45.000000:8.485281pt and 8.485281pt) -- +(45.000000:8.485281pt and 8.485281pt) -- +(135.000000:8.485281pt and 8.485281pt) -- +(225.000000:8.485281pt and 8.485281pt) -- cycle;
\draw (-251.000000, -345.000000) node {$H$};
\end{scope}
\draw[fill=white] (-257.000000, -375.000000) rectangle (-245.000000, -363.000000);
\draw[very thin] (-251.000000, -368.400000) arc (90:150:6.000000pt);
\draw[very thin] (-251.000000, -368.400000) arc (90:30:6.000000pt);
\draw[->,>=stealth] (-251.000000, -374.400000) -- +(80:10.392305pt);
\draw[fill=white] (-6.000000, -375.000000) rectangle (6.000000, -363.000000);
\draw[very thin] (-0.000000, -368.400000) arc (90:150:6.000000pt);
\draw[very thin] (-0.000000, -368.400000) arc (90:30:6.000000pt);
\draw[->,>=stealth] (-0.000000, -374.400000) -- +(80:10.392305pt);
\draw[color=black] (-251.000000,-387.000000) node[fill=white,below,minimum width=35.000000pt,minimum height=15.000000pt,inner sep=0pt] {\phantom{$x_{p,1}^K$}};
\draw[color=black] (-251.000000,-387.000000) node[below] {$x_{p,1}^K$};
\draw[color=black] (-0.000000,-387.000000) node[fill=white,below,minimum width=35.000000pt,minimum height=15.000000pt,inner sep=0pt] {\phantom{$x_{p,2}^K$}};
\draw[color=black] (-0.000000,-387.000000) node[below] {$x_{p,2}^K$};
\draw[dashed] (17.500000, -211.500000) -- (-268.500000, -211.500000);
\draw[draw opacity=0.000000,fill opacity=0.400000,fill=blue,rounded corners] (-268.500000,-201.000000) rectangle (17.500000,-222.000000);
\end{tikzpicture}


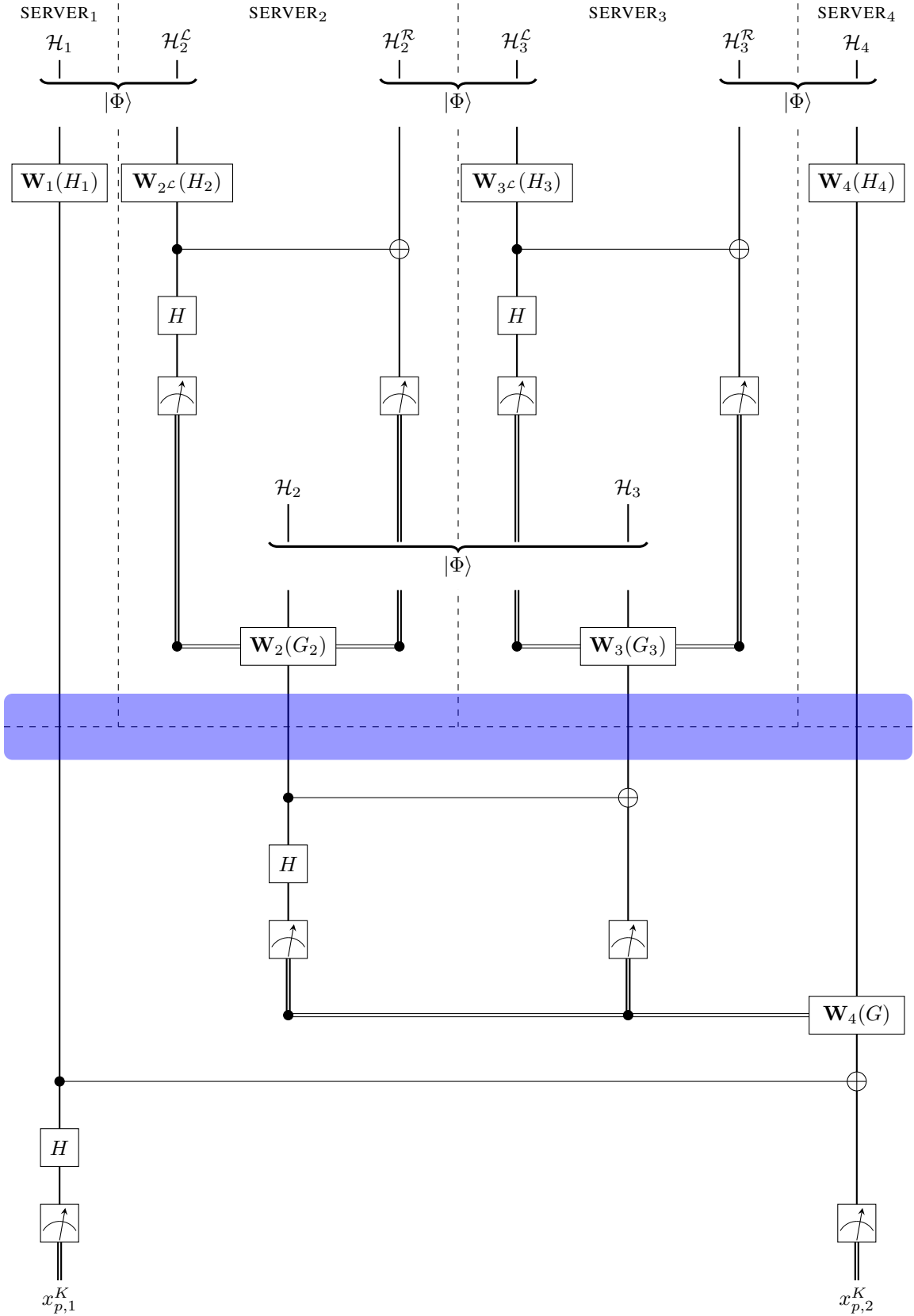
\captionof{figure}{Quantum circuit for the $[4,2]$-Reed--Solomon coded QPIR example. The blue rectangle represents the transmission of the four qubits from the servers to the user through 4 quantum channels.}

  \end{@twocolumnfalse}
]

\end{document}